\documentclass[11pt]{article}
\usepackage[utf8]{inputenc}
\usepackage[english]{babel}
\pdfoutput=1

\usepackage{geometry}
\usepackage{graphicx,graphics}

\usepackage{amsfonts,amssymb,amsmath,amscd,latexsym,authblk,amsthm}

\usepackage[citecolor=blue,linkcolor=red,colorlinks=true]{hyperref}

\usepackage{setspace}
\usepackage{natbib}
\usepackage{multirow}
\usepackage{color}
\usepackage{ulem}

\renewcommand{\P}{\mathbb{P}}
\newcommand{\Xall}{\mathbf{X}^{1:K}}
\newcommand{\bX}{\mathbf{X}}
\newcommand{\Xallij}{{X}_{ij}^{1:K}}
\newcommand{\Zall}{\mathbf{Z}}
\newcommand{\bZ}{\mathbf{Z}}
\newcommand{\bz}{\mathbf{z}}
\newcommand{\by}{\mathbf{y}}
\newcommand{\A}{\bf\mathcal{A}}

\newcommand{\bmu}{\boldsymbol{\mu}}
\newcommand{\balpha}{\boldsymbol{\alpha}}

\newcommand{\bpi}{\boldsymbol{\pi}}
\newcommand{\bbeta}{\boldsymbol{\beta}}

\newcommand{\btau}{\boldsymbol{\tau}}

\def \ind{\mathbb{I}}

\providecommand{\keywords}[1]{\textbf{\textit{Keywords---}} #1}

\newtheorem{Rem}{Remark}

\newtheorem{theorem}{Theorem}[section]

\author[1]{Pierre Barbillon \thanks{pierre.barbillon@agroparistech.fr}}
\author[1]{Sophie Donnet}
\author[2]{Emmanuel Lazega}
\author[3]{Avner Bar-Hen}

\affil[1]{AgroParisTech / UMR INRA MIA 518, 16 rue Claude Bernard, 75231 Paris Cedex 05, France}

\affil[2]{Institut d'\'Etudes Politiques de Paris (Sciences Po), D\'epartement de Sociologie,
Centre de Sociologie des Organisations,
19 rue Amélie, 75007
Paris, France}

\affil[3]{MAP5, UFR de Math\'ematiques et Informatique
Universit\'e Paris Descartes
45 rue des Saints-P\`eres
75270 Paris cedex 06}

\title{Stochastic Block Models for Multiplex networks: an application to networks of researchers}
\date{}
\begin{document}
\maketitle

 \begin{abstract}
Modeling relations between individuals is a classical question in social sciences and
clustering individuals according to the observed patterns of interactions allows to uncover a
latent structure in the data.
Stochastic block model (SBM) is a popular
approach for grouping the individuals with respect to their social comportment.  When several relationships of various types can occur jointly  between the individuals, the data are represented by multiplex networks where more than one edge can exist between the nodes.
In this paper, we extend the SBM
to multiplex networks in order to obtain a clustering based on more than one kind of relationship.
We propose to estimate the parameters --such as the marginal probabilities of
assignment to  groups (blocks) and the matrix of probabilities of connections between groups--  through a variational Expectation-Maximization procedure.
Consistency of the estimates as well as statistical properties of the model are obtained.
The number of groups is chosen thanks to the Integrated Completed Likelihood criteria, a penalized likelihood criterion.
Multiplex Stochastic Block Model arises in many situations but our applied example is motivated by a network of French cancer
researchers. The two possible links (edges) between researchers are a direct connection or a connection through their labs.
Our results show strong interactions between these two kinds of connections and the groups  that are obtained are discussed
to emphasize the common features of researchers grouped together.

 \end{abstract}


 \section{Introduction}

Network analysis has emerged as a key technique for understanding and for
investigating social interactions through the properties of relations between and within units.
From a statistical point of view, a network is a realization of a random graph formed by a set of nodes $V$ representing the units (e.g. individuals, actors, companies)
and a set of edges $E$ representing relationships between pairs of nodes. 
\\
The system in which the same nodes belong to multiple networks is typically referred to as a multiplex network or multigraph
(see \cite{wasserman1994social} for example).
In recent literature, there has been an upsurge of interest in multiplex networks (see for example \cite{PhysRevE.86.036115,2014arXiv1406.2205L,rank2010structural,szell2010multirelational,mucha2010community,maggioni2013multiplexity,brummitt2012suppressing,saumell2012epidemic,PhysRevE.87.062806,nicosia2013growing}).
In these multiplex networks, different kinds of links (or connections)  are possible for each pair of nodes. This induced
link multiplexity is a fundamental aspect of social relations \citep{snijders2003multilevel} since these multiple links
are frequently interdependent: 
links in one network may have an  influence on the formation or dissolution of links in other networks.
\\

The simultaneous analysis  of several networks   also arises  when one is interested in  the social comportment of individuals  belonging to organized entities (such as companies,
laboratories, political groups, etc.),
with some individuals possibly belonging to the same institution. While the actors will exchange resources (such as advice for instance)
at the individual level, their respective organizations of affiliation will also share resources at the institutional level
(financial resources for instance). Each level (individuals and organizations) constitutes a system of exchange of different
resources that has its own logic and could  be studied separately. However, studying the two networks jointly
(and hence embedding the individuals in the multilevel relational and organizational structures  constituting the inter-organizational
context of their actions) would allow us to identify the individuals that benefit from
relatively easy access to the resources circulating in each level, which is of much more interest. In other words, studying the two levels
jointly could help us understand how an individual can benefit from the position of its organization in the institutional network.

In this paper, we are interested in 
studying  the advice relations between researchers and the  exchanges of resources between laboratories. We adopt the following individual-oriented strategy (this point is discussed in the paper):
the institutional network is used to define a new network on the individual level i.e.  the set of nodes consists in the set of individuals and for a pair of individuals, two kinds of links are possible: a direct connection given by the individual network and a connection through their organizations given
by the organizational network. As a consequence,  the individual and institutional levels are fused into a multiplex.
\\

We then develop  a statistical model able to detect in multiplex
 substantial non-trivial topological features, with patterns of connection between their elements that are not purely regular.
Several models such as scale-free networks and small-world networks have been proposed to describe and understand the heterogeneity observed in networks.
These models  allow to derive properties of the network at the macro-scale and to understand   the outcomes of interactions. To explore heterogeneity at others scales (such as micro or meso-scale) in social networks,
specific models such as the stochastic block models (SBM) (\cite{Snijdersnowicki1997})
 have been developed for uniplex networks.
 In this paper, we propose an original extension of the SBMs to the multiplex case.  Our model is efficient to model
not only the main effects (that correspond to a classical uniplex) but also the pairwise interactions between the nodes. We estimate the parameters of the multiplex SBM using an extension of the variational EM algorithm.
Consistency of the estimation of the parameters is proved. As for uniplex SBM, a key issue is to choose the number of blocks.
We use a penalized likelihood criterion, namely Integrated Completed Likelihood (ICL).
 The inference procedure is performed on the cancer researchers / laboratories dataset.  \\

The paper is organized as follows. The extension of SBM to multiplex network is presented in Section
\ref{sec=multiplex_model}, the proofs of model identifiability and the consistency of variational EM procedure
are postponed in Appendices \ref{app:identif} and \ref{appendix_VarEM}.
In Section  \ref{sec=multilevel_result}, we describe  \cite{Lazega2008}'s dataset, apply the  new modeling and discuss the results. 
Eventually, the contribution of multiplex SBM to the analysis of multiplex networks is highlighted in Section \ref{sec=discussion}.

\section{Multiplex stochastic block model}\label{sec=multiplex_model}

The main objective is to  cluster the individuals (or nodes)
into blocks sharing connection properties with the other individuals of the multiplex-network.  Stochastic block models \citep{nowickiSnijders2001}
for random graphs have emerged as a natural tool to perform such a clustering based on uniplex networks (directed or not, valued or not).
In the following, we propose an extension of the Stochastic Block Model (SBM) to multiplex networks. The SBM for multiplex networks is derived from a
multiplex
Erd\"os-R\'enyi model which is described in subsection \ref{subsec:ER multilevel}.
The SBM for multiplex networks is derived in subsection \ref{subsec :SBM multilevel}.

\subsection{ Erd\"os-R\'enyi model for multiplex networks}\label{subsec:ER multilevel}

Let $\bX^1, \dots, \bX^K$ be $K$ directed graphs relying on the same set of nodes $E=\{1, \dots, n\}$.
We assume that $\forall (i,j), i\neq j, \forall k \in \{1,\dots,K\}$, $X_{ij}^k \in \{0,1\}$ and $X_{ii}\neq 0$. We define a joint distribution on $\Xall = (\bX^1, \dots,\bX^K)$ as:
$\forall (i,j) \in \{1,\dots, n\}^2, i\neq j$, $\forall w \in \{0,1\}^K$,
\begin{eqnarray}
\label{model:multi}
\P(\Xallij=w) =  \pi^{(w)} \mbox{ where }  \sum_{  w \in \{0,1\}^K} \pi^{(w)} =1\,,
\end{eqnarray}
and $(\Xallij)_{i,j}$ are mutually independent.

The maximum likelihood estimate of the parameter of interest  $\bpi =(\pi^{(w)})_{w \in \{0,1\}^K} $   is,
 for all $w \in \{0,1\}^K$:
\begin{equation}
\nonumber
 \widehat{\pi}_{w} = \frac{1}{n(n-1)} \sum_{i,j, i\neq j} \ind_{\{X_{ij}^{1:K} =w\}}\,.
\end{equation}

This model is quite simple since any relation between two individuals
(a relation being a collection of edges) does not depend on the relations between the other individuals.
However, the different kind of relations between two individuals (edges) are not assumed to be independent.

\begin{Rem} This  model is clearly an extension of the Erd\"os-R\'enyi model since the marginal
distribution of $X^{k}_{ij}$ (for any $k=1\dots K$) is Bernoulli with density :
\begin{equation}
\nonumber
 \P(X^{k}_{ij}=x^{k}_{ij})=\left(\sum_{w \in \{0,1\}| w_k=1}  \pi^{(w)}\right)^{x^{k}_{ij}}\left(\sum_{w \in \{0,1\}| w_k=0}  \pi^{(w)}\right)^{1-x^{k}_{ij}}\,.
\end{equation}
Moreover, any conditional distribution of $X^k_{ij}$ given $(X^l_{ij})_{l \in \mathcal{S}_{\backslash k}}$ (where $\mathcal{S}_{\backslash k}$ is a subset of $\{1,\dots,K\}$ not containing $k$)  is also univariate Bernoulli.
For instance, if $K=2$ the conditional distribution of $X^1_{ij}$ given $X^2_{ij}$ is
\begin{equation}
 \nonumber
 \P(X^1_{ij}=x^1_{ij}|X^2_{ij}=x^2_{ij})=\left(\frac{\pi^{(1,x^2_{ij})}}{\pi^{(1,x^2_{ij})}+\pi^{(0,x^2_{ij})}}\right)^{x^1_{ij}}\left(\frac{\pi^{(0,x^2_{ij})}}{\pi^{(1,x^2_{ij})}+\pi^{(0,x^{2}_{ij})}}\right)^{(1-x^1_{ij})}\,.
\end{equation}
Moreover, the  components of the bivariate Bernoulli random vector $(X^1_{ij},X^2_{ij})$ are independent if and only if $\pi^{(00)}\pi^{(11)} = \pi^{(10)}\pi^{(01)}$.

\end{Rem}

\paragraph{Introduction of explanatory variables.}
Naturally, the Erd\"os-R\'enyi model for multiplex networks can be extended to take into account explanatory variables.
Let $\by_{ij}$ denote the covariates characterizing the couple of nodes $(i,j)$, the model is defined by the
probabilities:
\begin{equation}\label{model1_mER_cov}
\begin{array}{ccl}
  \P(\Xallij= w)& =& \frac{\exp\left(\mu^{(w)}+ \left(\beta^{(w)}\right)^{\intercal} \by_{ij}\right)}{1+\sum_{v\not=(0,\ldots,0)}\exp\left(\mu^{(v)}+ \left(\beta^{(v)}\right)^{\intercal} \by_{ij}\right)}
  \quad \forall w\not=(0,\ldots,0),\\
 \P(\Xallij= (0,\ldots,0))& =&\frac{1}{1+\sum_{v\not=(0,\ldots,0)}\exp\left(\mu^{(v)}+ \left(\beta^{(v)}\right)^{\intercal} \by_{ij}\right)}  \,,
 \end{array}
 \end{equation}
where $x^{\intercal}$ denotes the transposed vector of $x$.

\begin{Rem}
Note that in the multiplex Erd\"os-R\'enyi model, the modeling is actor based, which means that the individuals are the same for all
the networks $X^1, \dots,X^K$ and   we model conjointly all the connections. As a consequence, the covariates $\by_{ij}$ only depend on the couple $(i,j)$
and are not linked to the network under consideration.
\end{Rem}

Since the multiplex Erd\"os-R\'enyi model belongs to exponential models,  the generalized linear model theory applies when we introduce the covariates as in
model (\ref{model1_mER_cov}) and the estimates are obtained using standard optimization strategies.

\subsection{Stochastic block model for multiplex networks}\label{subsec :SBM multilevel}
When the goal is to cluster  individuals according to  their social comportment, we can
derive a Stochastic Block Model  version of the multiplex Erd\"os-R\'enyi model.
Let $Q$ be the number of blocks and $Z_i$ the latent variable such that $Z_i=q$ if  the individual $i$ belongs to block $q$
(note that an individual can only belong to one block in this version).

\vspace{1em}

The multiplex version of SBM is written  as follows: $\forall (i,j) \in \{1,\dots,n\}^2, i \neq j$, $\forall w\in \{0,1\}^{K}$, $\forall (q,l) \in \{1,\dots,Q\}^2$,
\begin{equation}\label{modelK}
\begin{array}{rcl}
\P(X^{1:K}_{ij} = w|Z_i=q,Z_j=l)
&=& \pi_{ql}^{(w)}   \\
P(Z_i = q) &=& \alpha_q\,.
\end{array}
\end{equation}
Such a model  includes $(2^K-1)Q^2+(Q-1)$ parameters. Introducing the following notations:
\begin{equation}
\nonumber
\balpha = (\alpha_1,\dots,\alpha_Q), \quad \bpi = (\pi^{(w)}_{ql})_{w\in \{0,1\}^K, (q,l) \in \{1,\dots,Q\}^2},\quad  \theta=(\balpha,\bpi)\,,
\end{equation}
the likelihood function is written as:
\begin{eqnarray}\label{eq:likelihood SBM-Klevel}
 \ell(\Xall; \theta) &=& \int_{\bz  \in \{1,\dots Q\}^n} p(\Xall | \bZ;\bpi) p( \bZ; \balpha) d \bZ \nonumber\,,\\
 &=& \sum_{\bZ  \in \{1,\dots Q\}^n}  \prod_{i,j,i\neq j}   \pi_{Z_i Z_j}^{(X^{1:K}_{ij})}  \prod_{i=1}^n \alpha_{Z_i}\,,
\end{eqnarray}
%$\theta=(\balpha,\bpi)$
where the latent variable (the block affectations) are integrated out.
The identifiability of the model can be proved (see Appendix \ref{app:identif}, theorem \ref{theo:identif}) and the maximum likelihood estimators are consistent (theorem \ref{theo:consist}).

\begin{Rem}
 Note that, as before, covariates on the couple $(i,j)$ can be introduced in the model: $\quad \forall w\not=(0,\ldots,0),$
\begin{equation}\label{eq:MulitSBM_cov}
\begin{array}{ccl}
  \P(\Xallij= w|Z_i=q,Z_j=l)& =& \frac{\exp\left(\mu_{ql}^{(w)}+ \left(\beta_{ql}^{(w)}\right)^{\intercal} \by_{ij}\right)}{1+\sum_{v\not=(0,\ldots,0)}\exp\left(\mu_{ql}^{(v)}+ \left(\beta_{ql}^{(v)}\right)^{\intercal} \by_{ij}\right)} \,,\\
 \P(\Xallij= (0,\ldots,0)|Z_i=q,Z_j=l)& =&\frac{1}{1+\sum_{v\not=(0,\ldots,0)}\exp\left(\mu_{ql}^{(v)}+ \left(\beta_{ql}^{(v)}\right)^{\intercal} \by_{ij}\right)}  \,.
 \end{array}
  \end{equation}
However, the number of parameters in this new model will increase drastically, leading to estimation issues.
\end{Rem}

\paragraph{Maximum likelihood and model selection}
As soon as $n$ or $Q$ are large, the observed likelihood (\ref{eq:likelihood SBM-Klevel})
is not tractable (due to the sum on $\bZ  \in \{1,\dots Q\}^n$) and its maximization is a challenging task.
Several approaches have been developed in the literature \citep[for a review, see][]{matias:hal-00948421},
both in the frequentist  and  Bayesian frameworks, starting from \cite{Snijdersnowicki1997} and \cite{nowickiSnijders2001}.
However, when the latent data space is really large, these techniques can be burdensome.
Some other strategies have been proposed, such as  \cite{bickelchen2009} which relying  on a profile-likelihood optimization
or  the moment estimation proposed by \citet{ambroisematias2012},  to name but a few .

The variational EM in the context of SBM proposed by \cite{Daudinetal2008} is a flexible tool to tackle the computational challenge in many types of graphs.
Simulation studies showed its practical efficiency \citep{mariadassou2010}.
Moreover, its theoretical convergence towards the maximum likelihood estimates has been studied by  \cite{celisse:daudin:laurent:2012} for binary graphs.   In this paper, we adapt the variational EM to multiplex networks.
The algorithm is  described in  Appendix \ref{appendix_VarEM} and its convergence towards the true parameter is proved (Theorem \ref{theo:conv var EM}).

The selection of the most adequate number of blocks $Q$ is performed using a modification of the ICL criterion  \citep[as in][]{Daudinetal2008,mariadassou2010}.
Let $\mathcal{M}_Q$  denote the model defined  (\ref{eq:MulitSBM_cov}) with $Q$ blocks.
\begin{eqnarray*}
ICL(\mathcal{M}_Q)& =& \max_{\theta} \log p(\Xall, \widetilde{\bZ}; \theta) -   \frac{1}{2}\left\{ Q^2 (2^K-1)  \log(K n (n-1)) +   (Q-1) \log n \right\}\,.
\end{eqnarray*}
where $\widetilde{\bZ}$ are the predictions of the assignments $\bZ$ obtained as a sub-product of the variational EM algorithm (see Appendix \ref{appendix_VarEM}).
As in the BIC criteria, the $\log$ refers to the number of data. Thus, the $n$ nodes are used to estimate the $Q-1$ probabilities $\alpha_1, \dots, \alpha_{Q-1}$. The $K n (n-1)$ edges are used to estimate $\bpi$.  No theoretical results exist for the ICL properties, but this criteria has proved its efficiency in practice.

\begin{Rem}
If we consider covariates as in equation (\ref{eq:MulitSBM_cov}), then the ICL is adapted :
\begin{equation}
 \nonumber
 ICL(\mathcal{M}_Q) = \max_{\balpha,\bmu,\bbeta} \log p(\Xall, \widetilde{\bZ};\balpha,\bmu,\bbeta)  - \frac{1}{2}\left\{ P_Q \log(K n (n-1)) +   (Q-1) \log n \right\} \,,
\end{equation}
where  $P_Q$ is the dimension of $(\bmu, \bbeta)$.
\end{Rem}

% !TEX root = test2.tex

\section{Analysis for the laboratory-researcher data}\label{sec=multilevel_result}

\subsection{The data}

French  scandals during the 1990s  involving the voluntary sector around the cancer research dried up large donations that funded research laboratories.
In the 2000s, the cancer research became politicized, with the launch of the Cancer Plan and the creation of a dedicated institution.  The aim of this public  agency is to coordinate the cancer research and to promote collaborations about top researchers.
In this context, \cite{Lazega2008} studied the relations of advice between French cancer researchers  identified as ``Elite'' conjointly with the relations of their respective laboratories.

At the inter-individual level, the actors  (researchers) were submitted a list of cancer researchers and asked in interviews whom they sought advice from.  The advice were of five types,
namely  advice to deal with choices about the direction of projects, advice to  find institutional support, advice to handle financial
resources, advice for  recruitment, and finally advice about manuscripts
before submitting them to journals. The five advice networks are too sparse to be studied separately. That is why  they were aggregated. Therefore two researchers are considered as linked if at least
one kind of relationship exists.
Obviously the links are directed.\\
At the laboratory level (concerning only laboratories with ``Elite'' researchers),
the laboratory directors were asked to specify what type of resources they exchanged
with the other laboratories on the list. The examined resources were the recruitment
of post-docs and researchers, the development of programs of
joint research, joint responses to tender offers, sharing of technical
equipment, sharing of experimental material, mobility of
administrative personnel, and invitations to conferences and seminars.
Once again, to avoid over-sparsity,  the various networks were aggregated and two laboratories are said to be linked if there is at least one link between them.
From this network on labs, we can derive indirect links between the researchers, i.e. two researchers are connected if their laboratories exchange resources.
We finally have two adjacency matrices on the same set of nodes (researchers).

This corresponds to transforming the multilevel network (individual/ organization) into a multiplex network (several kinds of relations among individuals).  This is reasonable since the
majority of laboratories contains a unique ``Elite'' researcher. Thus, there is no big difference in the number of nodes between
the institutional and individual  levels.

In addition, auxiliary covariates are available to describe  the researchers:
 their age, their specialty, two publication performance scores based on two
 periods of five years, their status (director of the lab or not).
Auxiliary covariates are also available the laboratories: their size (number of researchers) and
 their location. \\
 Complete data for $95$ researchers identified as the ``Elite'' of French cancer research working in $76$ laboratories
are available. \\

\subsection{Statistical inference through SBM for multiplex}

To estimate the parameters of the multiplex SBM model on this dataset, we use a modified version of the variational EM algorithm \citep{Daudinetal2008} described in the Appendix Section \ref{appendix_VarEM}.
The optimization of the ICL criterion derived from likelihood (\ref{eq:likelihood SBM-Klevel})
 leads to four blocks (indistinctly denominated clusters or groups).
 For the sake of clarity, we index by $R$ and $L$ (rather than $\bX^{1}$, $\bX^{2}$) the two adjacency matrices (respectively the direct and indirect ones)

 In Figures \ref{figRconnection} and \ref{figLconnection} we plot the marginal and conditional probabilities
 of the connections of researchers (respectively labs) between and within blocks. 

Note that the study of the estimated  marginal distributions allows us to have results on the researchers without considering the laboratories.
This gives a clear interpretation of the importance of the lab for the researcher network structure. The obtained blocks are described in Table
\ref{tabmulti}: (\ref{tabmulti}a) gives the sizes of the four blocks;
(\ref{tabmulti}b), (\ref{tabmulti}c) and (\ref{tabmulti}d) describe the blocks with respect to the covariates ``location", ``director of not'' and ``specialty''.
The estimations by the variational EM procedure
 were conducted by wmixnet \citep{leger2014} with our specific implementation of the bivariate Bernoulli model. \\

\begin{figure}[h!t!]
\begin{center}
\begin{tabular}{cc}
\multicolumn{2}{c}{\includegraphics[width=7cm]{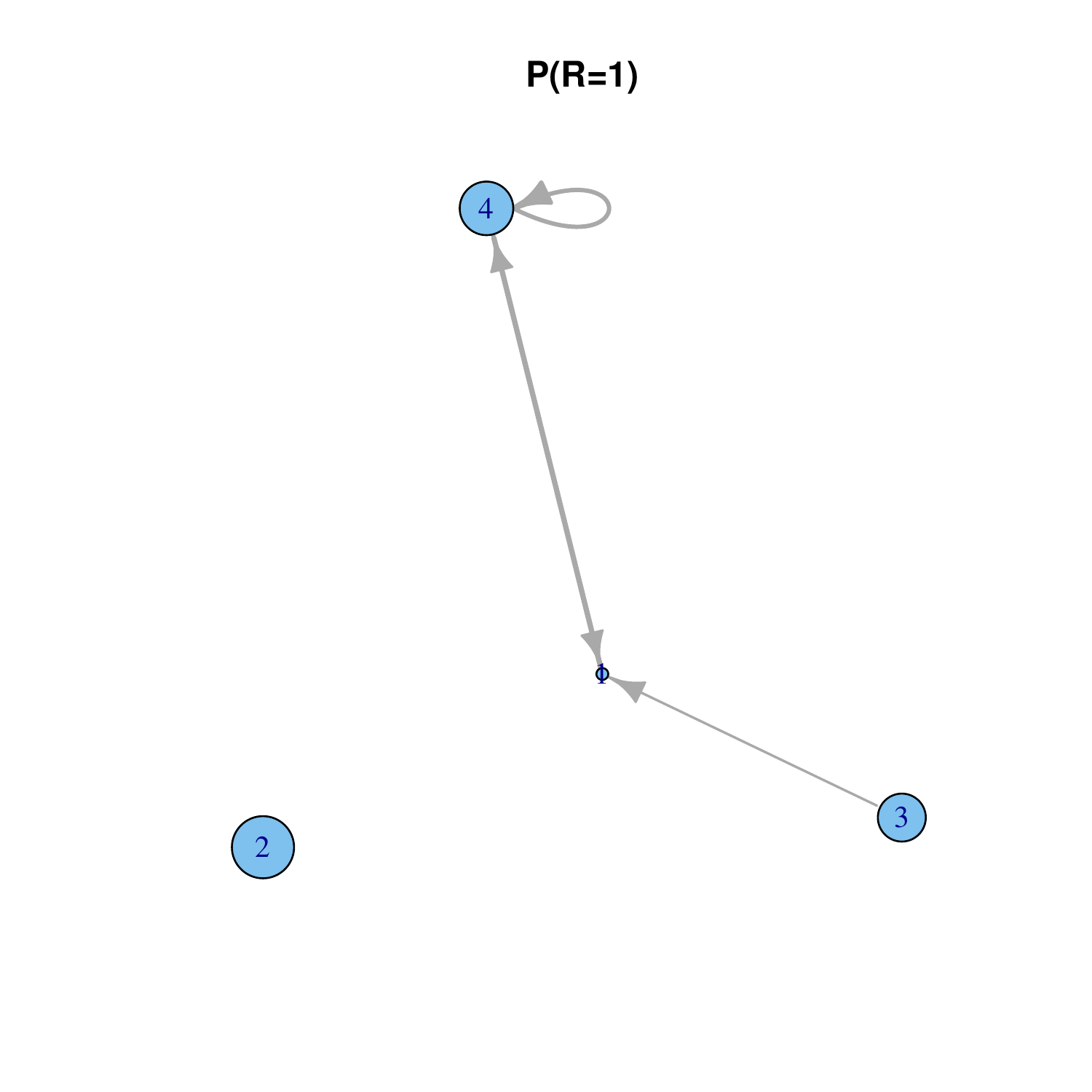}}\vspace{-2cm}\\
\includegraphics[width=7cm]{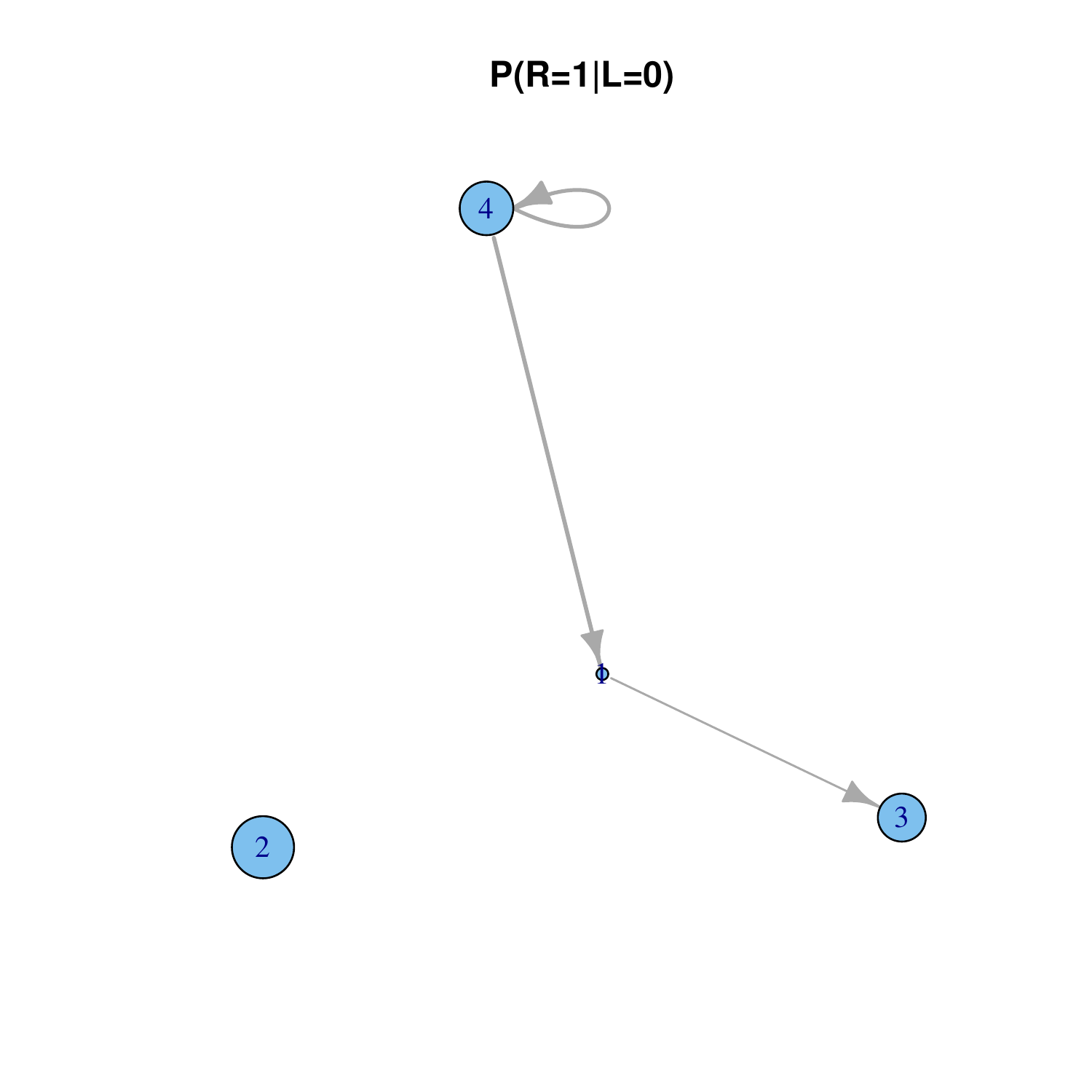}&\includegraphics[width=7cm]{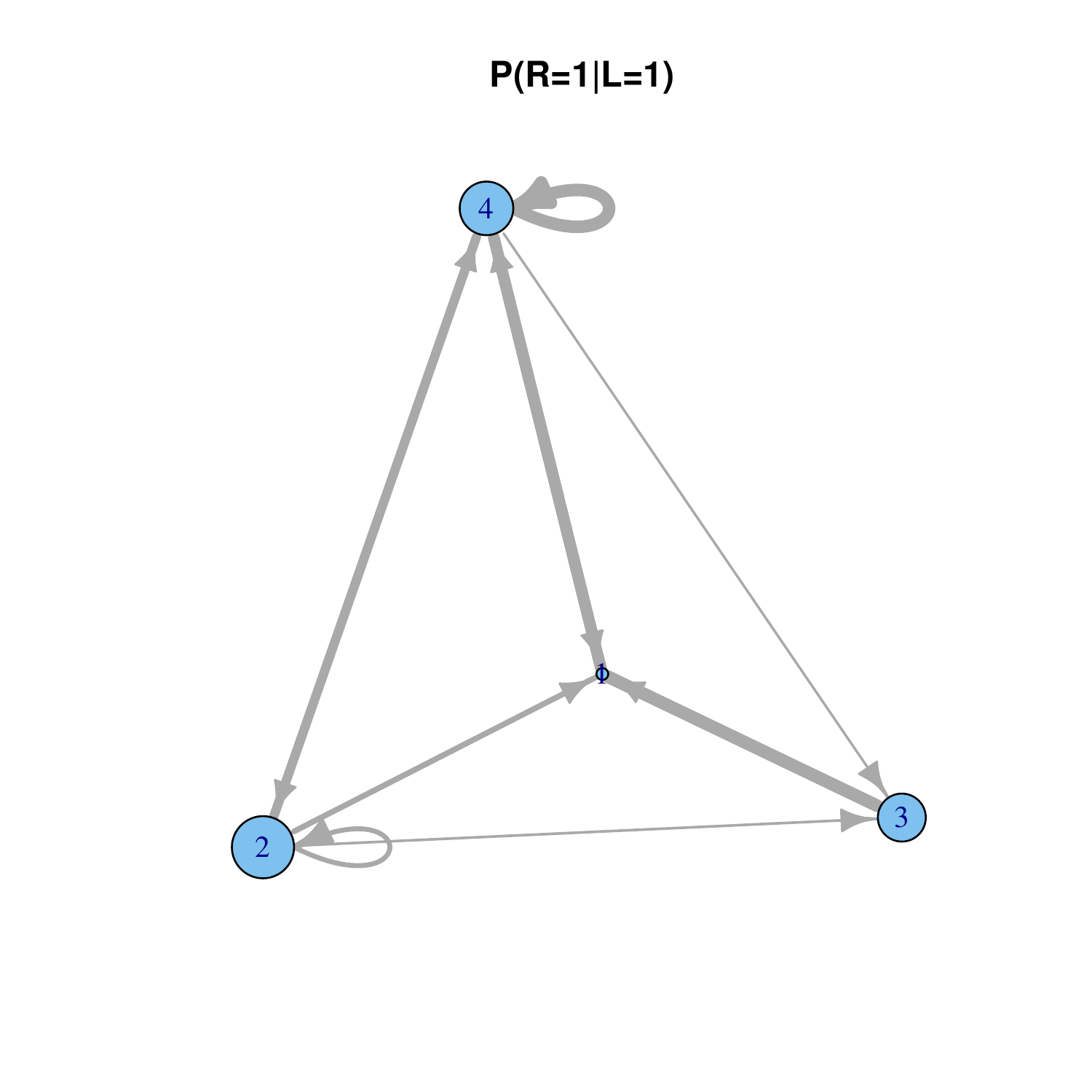}
\end{tabular}
 \end{center}
 \caption{Marginal probabilities of Researcher connections between and within blocks (top) and probabilities of Researcher
 connections between and within blocks
 conditionally on absence (bottom left-hand-side) or presence (bottom right-hand-side) of Lab connection.
 Vertex size is proportional to the block size. Edge width is proportional to the probabilities of connection; if this probability is
 smaller than 0.1, edges are not displayed.}
\label{figRconnection}
 \end{figure}

\begin{figure}[h!t!]
\begin{center}
\begin{tabular}{cc}
\multicolumn{2}{c}{\includegraphics[width=7cm]{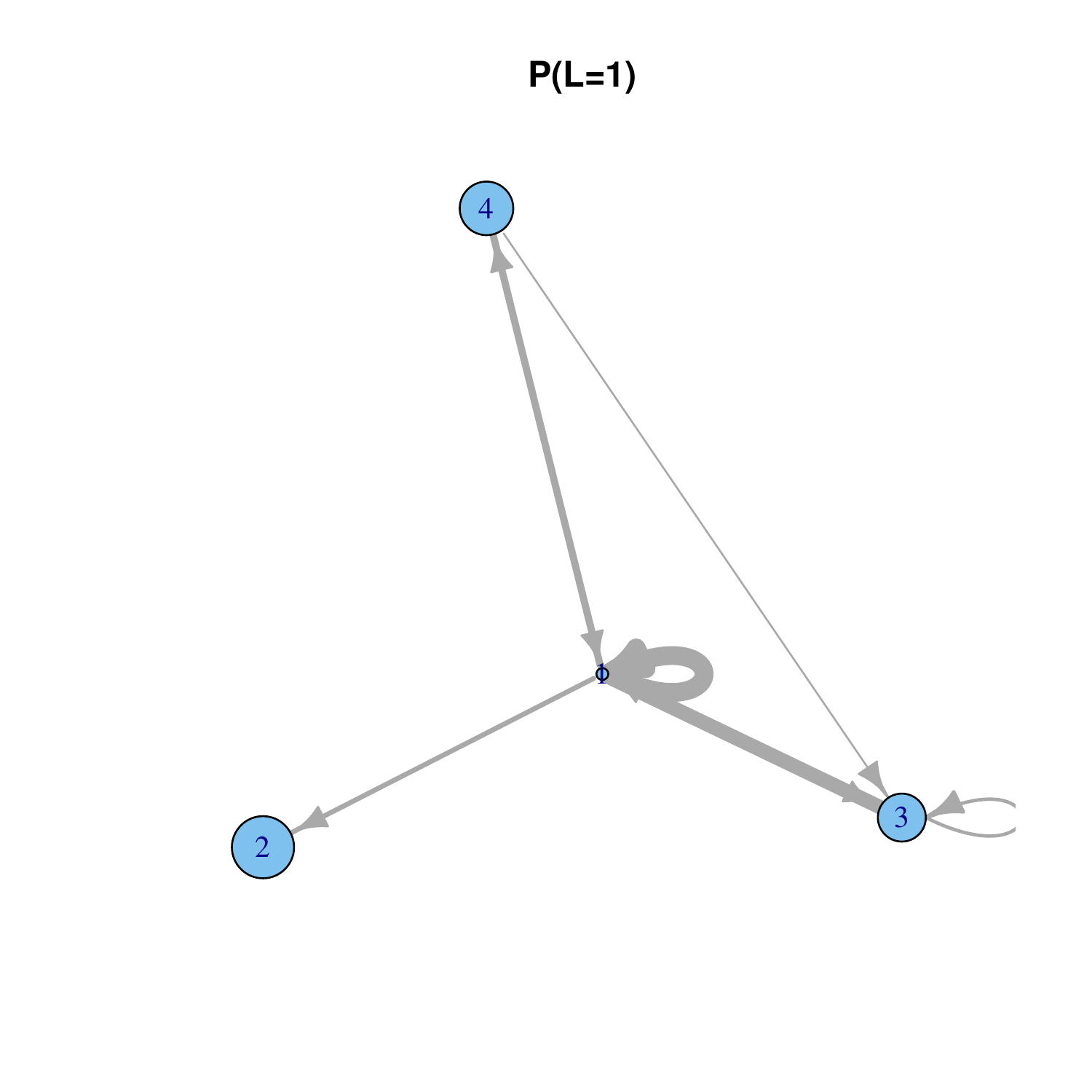}}\vspace{-2cm}\\
\includegraphics[width=7cm]{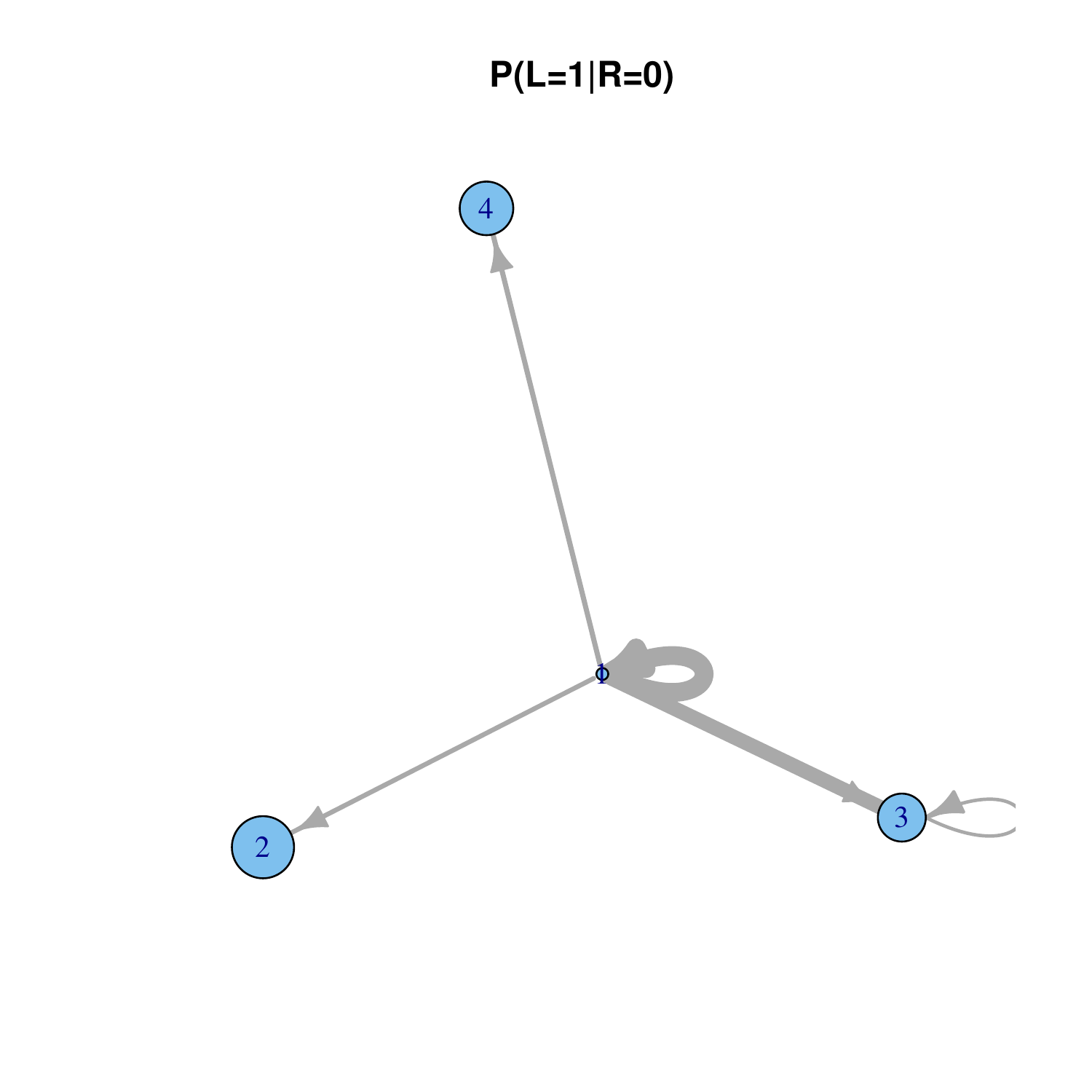}&\includegraphics[width=7cm]{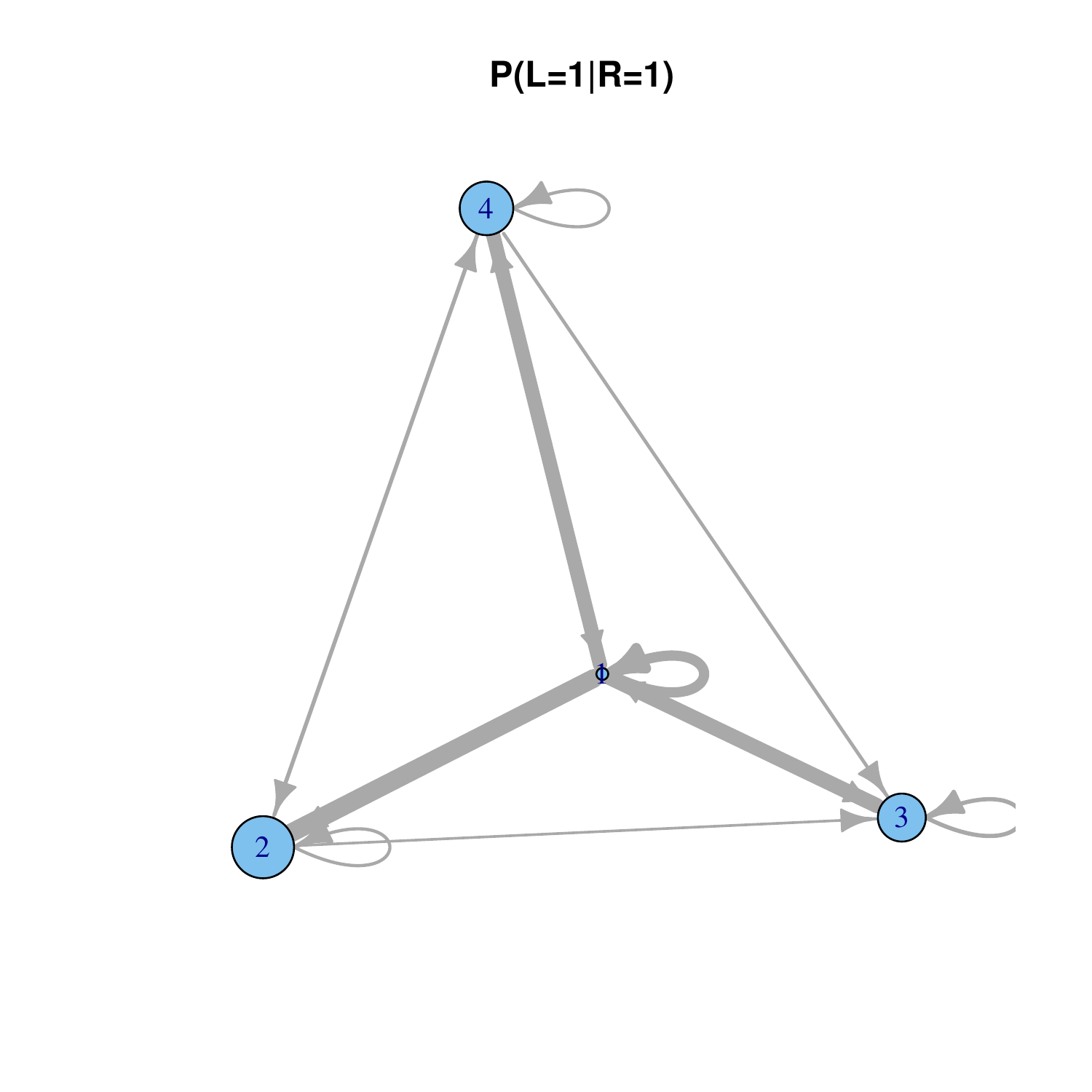}
\end{tabular}
 \end{center}
 \caption{Marginal probabilities of Lab connections between and within blocks (top) and probabilities of Lab connections between and within blocks
 conditionally on absence (bottom left-hand-side) or presence (bottom right-hand-side) of Researcher connections.
 Vertex size is proportional to the block size. Edge width is proportional to the probabilities of connection; if this probability is
 smaller than 0.1, edges are not displayed.}
\label{figLconnection}
 \end{figure}

We now discuss the results. Multiplex SBM reveals interesting structural features of the multiplex network.
More precisely, collaboration takes place in a clustered manner for both researchers and laboratories;
collaborating laboratories tend to have affiliated researchers seeking advice from one another.
Indeed, Figure \ref{figRconnection} shows that the existence of a connection (exchange of resources) between labs
clearly increases the probability of connection (sharing advice) between researchers.
The reinforcement of this probability of connection
is clearly outstanding in block 2. In this block, the researcher connections are
quite unlikely within the block or with other blocks. However, conditionally to the existence of a laboratory connection,
the researcher connections become more important especially with block 4.
%These researchers benefits from structure provided by their lab
In block  4, the links between researchers are strengthened given a connection  between their laboratories.
Researchers in block 3 seem to be the least affected by the connections provided by their laboratories.
The case of block 1 is quite peculiar  since it contains two researchers only.
This clustering demonstrates that not all researchers benefit on equal terms  from the institutional level. Some researchers
are more dependent on their laboratories in terms of connections.
Furthermore, Figure \ref{figLconnection} shows the likeliest connections between laboratories
are mainly with  block 1.  The fact that researchers  are sharing advice  inflates the probability of exchanging  resources between laboratories.

%\Sophie{Dans l'analyse revenir au sens que ce qu'est la connexion : conseil et echange de ressources}

\begin{table}
 \caption{\label{tabmulti}(a) Size of blocks obtained by multiplex SBM.
%(b) Crossing of assignments given by the SBM model with the ones given by \citet{Lazega2008} (SF: small fish, BF: big fish).
 (b) Cross frequencies of blocks versus the lab location (in \^{I}le-de-France (IdF) or not).
(c) Cross frequencies of blocks versus the researcher's status (lab director or not).
 (d) Cross frequencies of blocks versus the researcher's specialties (PH: Public Health;%code10
Su: Surgery; %code20
He: Hematology; %code21
ST: Solid Tumours; %code 22
FPh: Fundamental pharmacology; %code30
FMo: Fundamental molecular research; %code 40
FMoG: Fundamental molecular genetic research%code 50
).
}

\begin{tabular}{lll}
 (a)\fbox{
\begin{tabular}{rr}
 & block \\
  \hline
1 &   2 \\
  2 &  48 \\
  3 &  19 \\
  4 &  26 \\
\end{tabular}}
& \quad (b)
 \fbox{
 \begin{tabular}{rrr}
 & not idf & idf \\
  \hline
1 &   1 &   1 \\
  2 &  26 &  22 \\
  3 &  10 &   9 \\
  4 &  11 &  15 \\
\end{tabular}}
& \quad (c)
\fbox{
 \begin{tabular}{rrr}
 & not director &director \\
  \hline
1 &   1 &   1 \\
  2 &  21 &  27 \\
  3 &  12 &   7 \\
  4 &  12 &  14 \\
\end{tabular}}
\\
\\
\multicolumn{3}{c}{
(d)
\fbox{
\begin{tabular}{rrrrrrrr}
 & PH & Su & He & ST & FPh & FMo & FMoG \\
  \hline
1 &   0 &   0 &   0 &   0 &   0 &   1 &   1 \\
  2 &  12 &   5 &   6 &  10 &   7 &   7 &   1 \\
  3 &   0 &   1 &   3 &   1 &   1 &  10 &   3 \\
  4 &   6 &   1 &   7 &   4 &   2 &   2 &   4 \\
\end{tabular}}}
\end{tabular}
\end{table}

More importantly, some complex features of the within-level network structure are explained mainly by cross-level interactions.
In this empirical case, units of each level are clustered in blocks that make sense from the perspective of the categories of
people and organizations, at each level separately and together.

	Block 1 members  work in the biggest labs in terms of size (Figure \ref{figboxplot}(b)).
	SBM clusters them because they have many more relations than the other members of the network. They have among the highest indegrees,
	and average outdegrees in labs that have highest average indegrees and outdegrees (Figures \ref{figboxplot}(e), \ref{figboxplot}(f), \ref{figboxplot}(g), \ref{figboxplot}(h)).
	They have the highest performance in terms of  publication performance scores  in both periods (Figures \ref{figboxplot}(c) and \ref{figboxplot}(d)).
	In fact, they have a similar relational profile since they are providers
	of transgenic mice for the experiments of many similar colleagues.

	Block 2 members  are  among the lowest indegrees and outdegrees in labs that have the lowest average indegrees and outdegrees (Figures \ref{figboxplot}(e), \ref{figboxplot}(f), \ref{figboxplot}(g), \ref{figboxplot}(h)),
	slightly older than the others (Figure \ref{figboxplot}(a)), mainly in the smallest labs in terms of size.
	This block is heterogeneous in terms of specialties
	(especially 40\% clinicians and 25\% diagnostic/prevention/epidemiology specialists)
	except fundamental research (Table \ref{tabmulti}(d)).
	They also have among the lowest performance levels for both periods (Figures \ref{figboxplot}(c), \ref{figboxplot}(d)),
	although this is increasing.
	This is the biggest block. Their behavior may be described as fusional as proposed in \citet{Lazega2008} since
	it corresponds to individuals for whom the probabilities of connection are the most affected by
	the connections of their laboratories.

In block 3, younger fundamental researchers in laboratories carrying out fundamental research  (mostly molecular research) prevail
(Figure \ref{figboxplot}(a) and Table
\ref{tabmulti}(d)).  They have relatively low indegrees and outdegrees,
in labs that have the highest indegrees (after Block 1) and average outdegrees (Figures \ref{figboxplot}(e), \ref{figboxplot}(f), \ref{figboxplot}(g), \ref{figboxplot}(h)).
70\% are among the top performers of this population, i.e. highest performance levels after Block 1 members, for both periods (Figures \ref{figboxplot}(c), \ref{figboxplot}(d)).
Their dominant relational strategies are individualist or independentists (as shown in Figure \ref{figRconnection})
since
the probabilities of connection between researchers remain quite unchanged, no matter if their laboratories are connected.

			Block 4 is also heterogeneous in terms of specialties but its largest subgroup is composed of hematologists (clinical and fundamentalists)
	(Table \ref{tabmulti}(d)).
	 Researchers have average indegrees and outdegrees in laboratories that have relatively low indegrees
	and outdegrees and that are also of average size (Figures \ref{figboxplot}(e), \ref{figboxplot}(f), \ref{figboxplot}(g), \ref{figboxplot}(h)).
	There are proportionally
	more directors of laboratories in this block than in the others (Table \ref{tabmulti}(c)).
	Their dominant strategy can be called fusional or collectivist.
	The performance levels of the majority for both periods are somewhat mixed and average, but decreasing
	(Figures \ref{figboxplot}(c), \ref{figboxplot}(d)). In this block, researchers can take advantage of their laboratories to be connected with colleagues but
	they can also hold relations apart from their laboratories.

	In this case, SBM highlight in the  data a specific kind of block structure that provides an interesting understanding of the effect of dual positioning.
SBM mixes people and laboratories that previous analyses used to separate. It is interesting to notice, for example, that SBM partitions the population
regardless of geographic location (Table \ref{tabmulti}(b)) –a criterion that was previously shown to be meaningful to understand collective action in this research milieu.
In fact this partition may highlight the links between laboratories and between individuals that cut across
the boundaries and remoteness  created by geography, showing that certain categories of actors tend to reach across
these separations when it suits them, either as investments to prepare future collaborations,
or as a follow-up for past investments.

\begin{figure}
 \centering
\begin{tabular}{cc}
 \includegraphics[scale=.3]{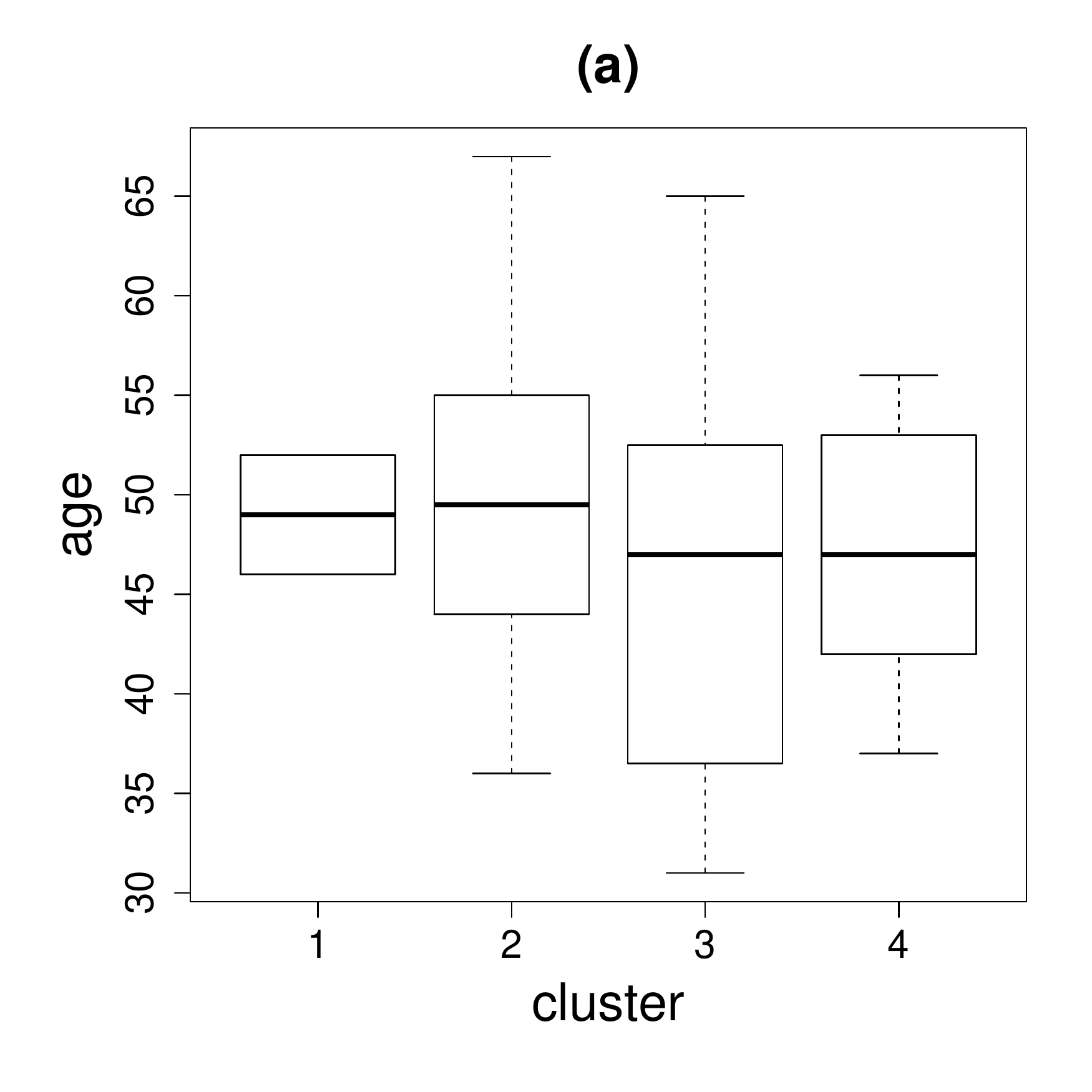} & \includegraphics[scale=.3]{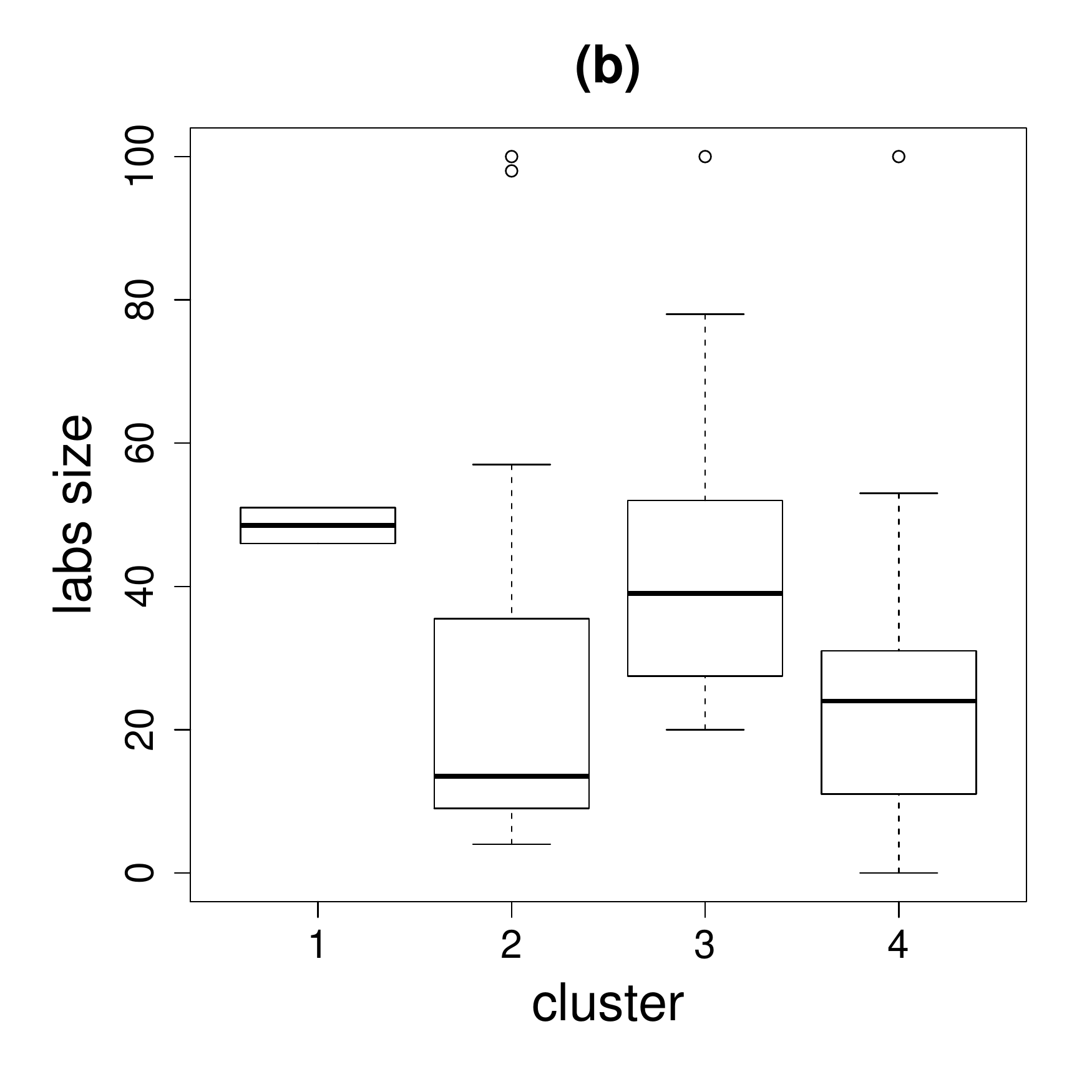}\\
 \includegraphics[scale=.3]{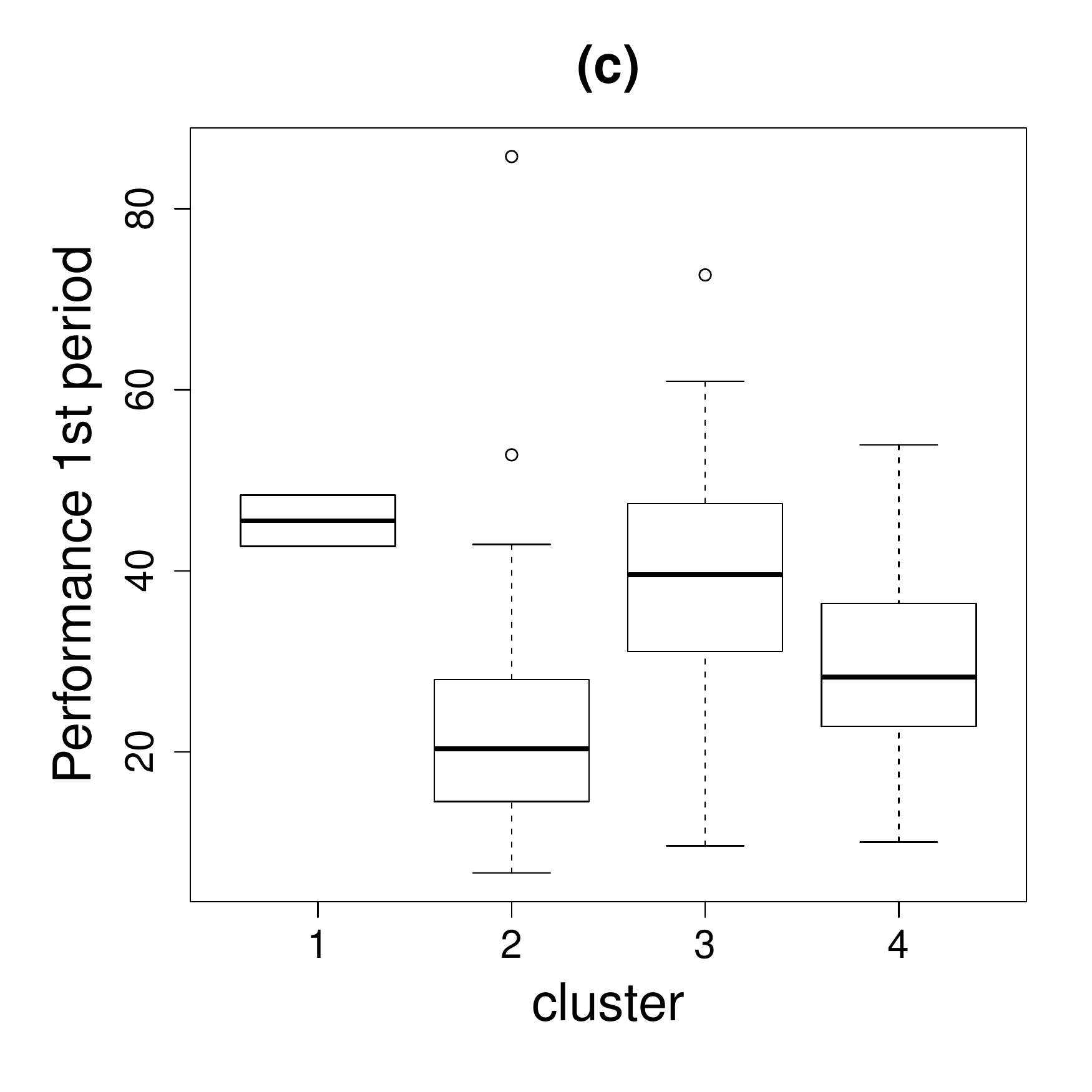}&\includegraphics[scale=.3]{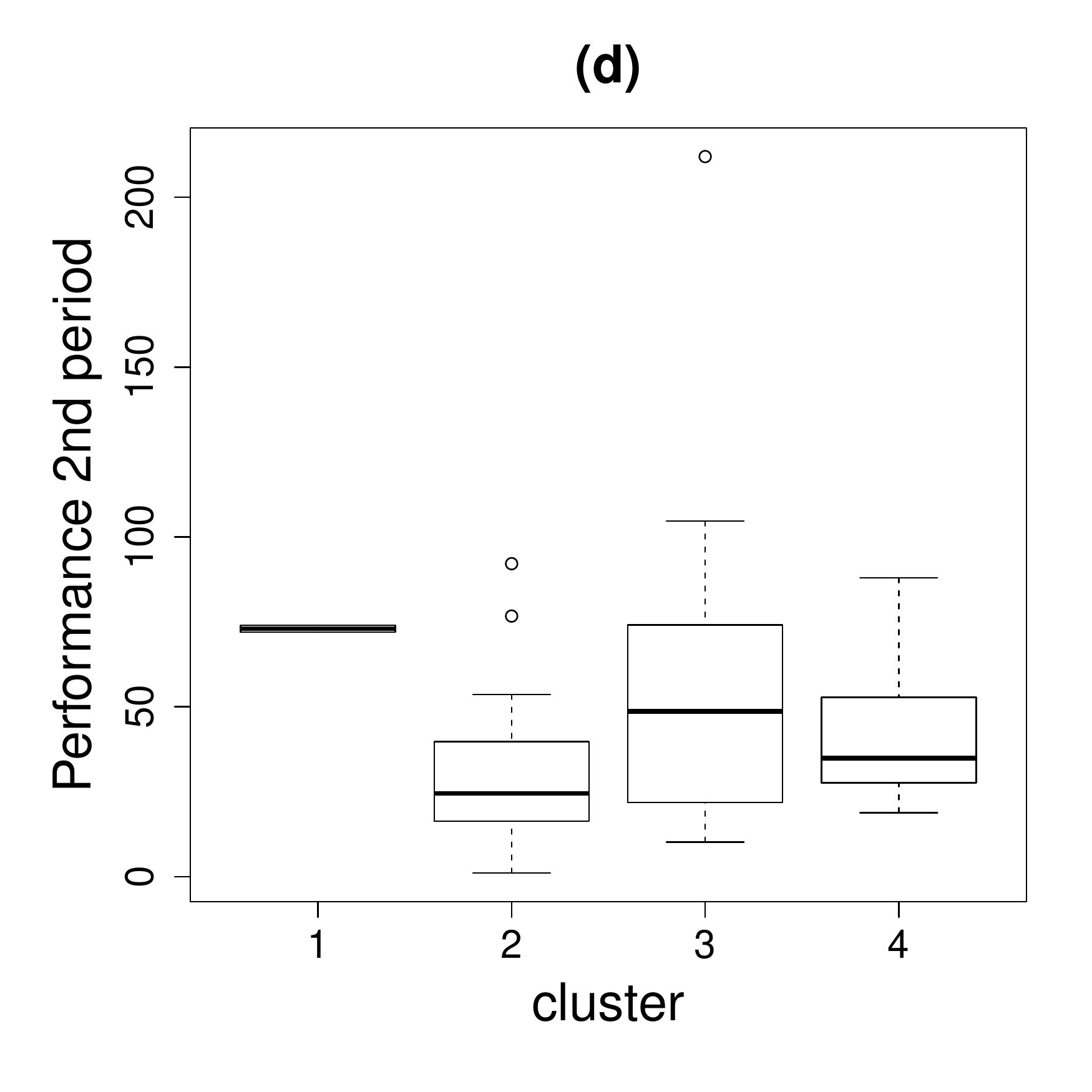}\\
   \includegraphics[scale=.3]{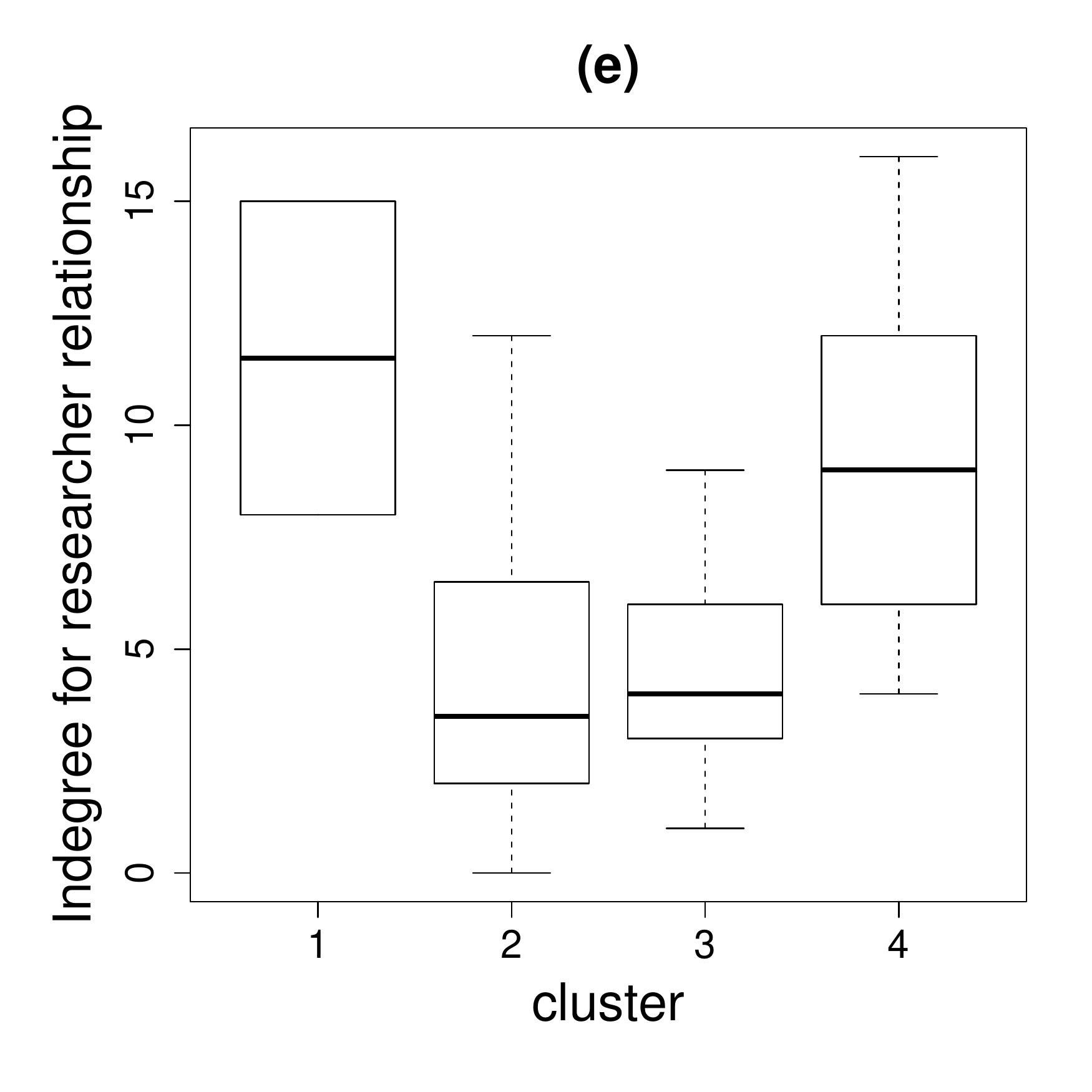}&\includegraphics[scale=.3]{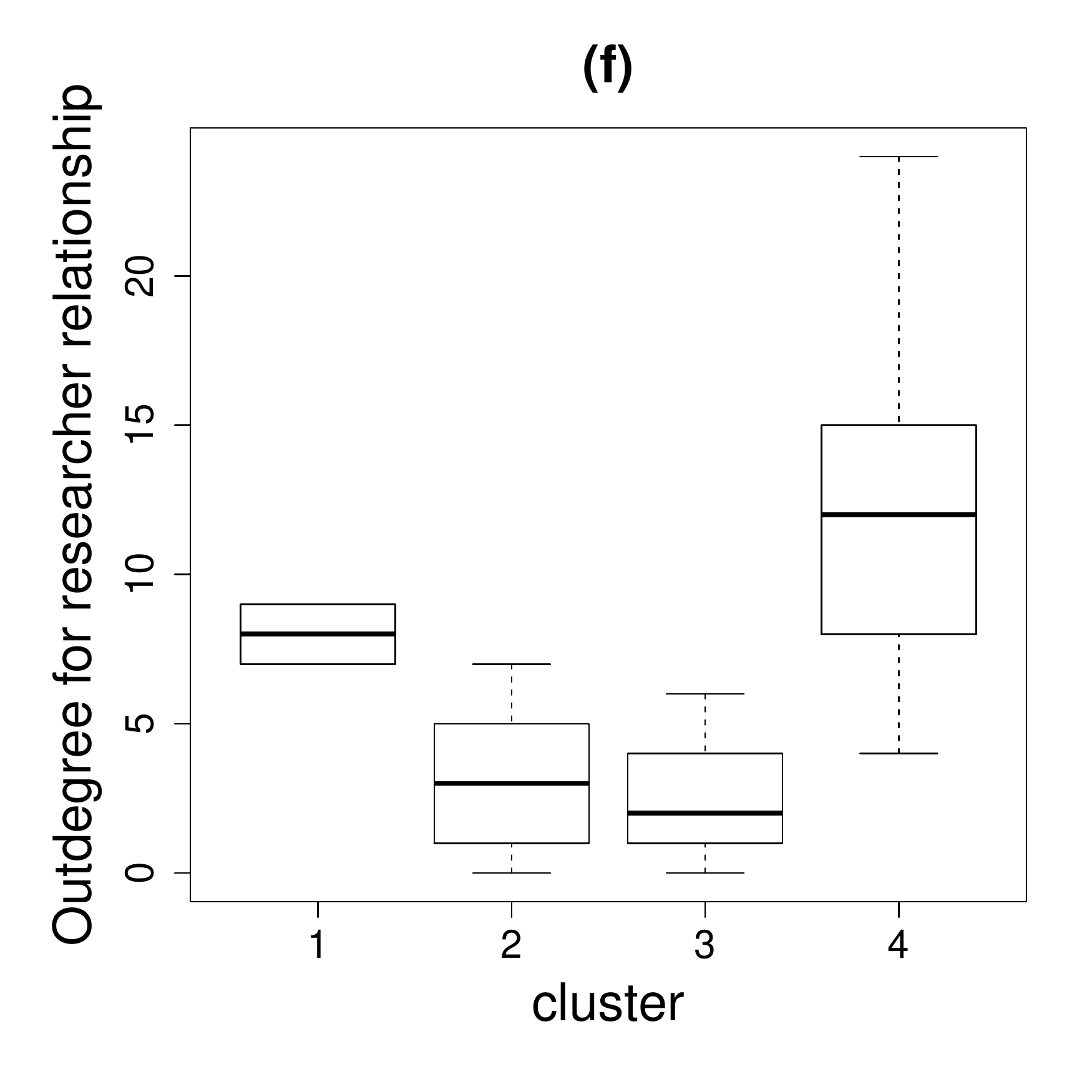}\\
   \includegraphics[scale=.3]{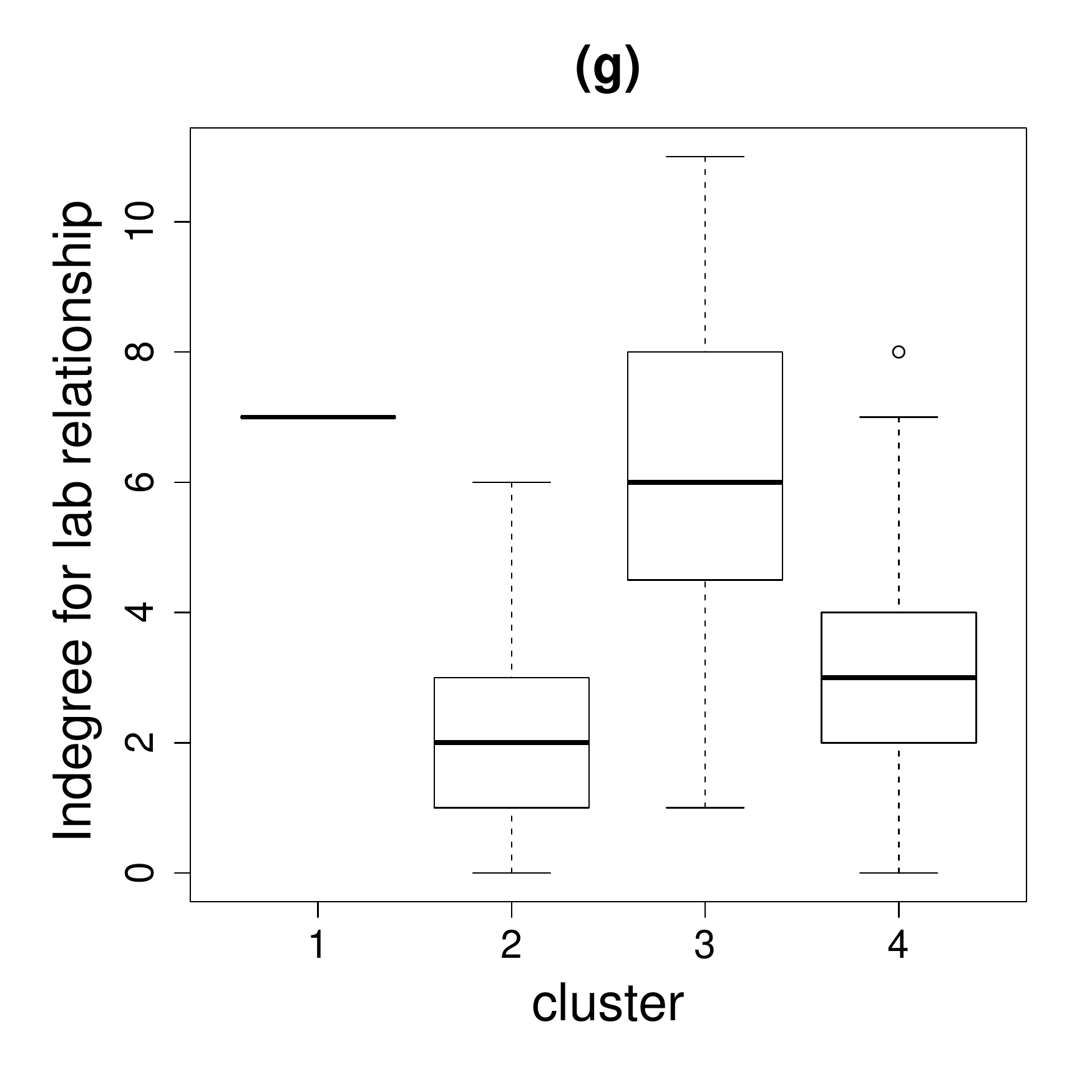}&\includegraphics[scale=.3]{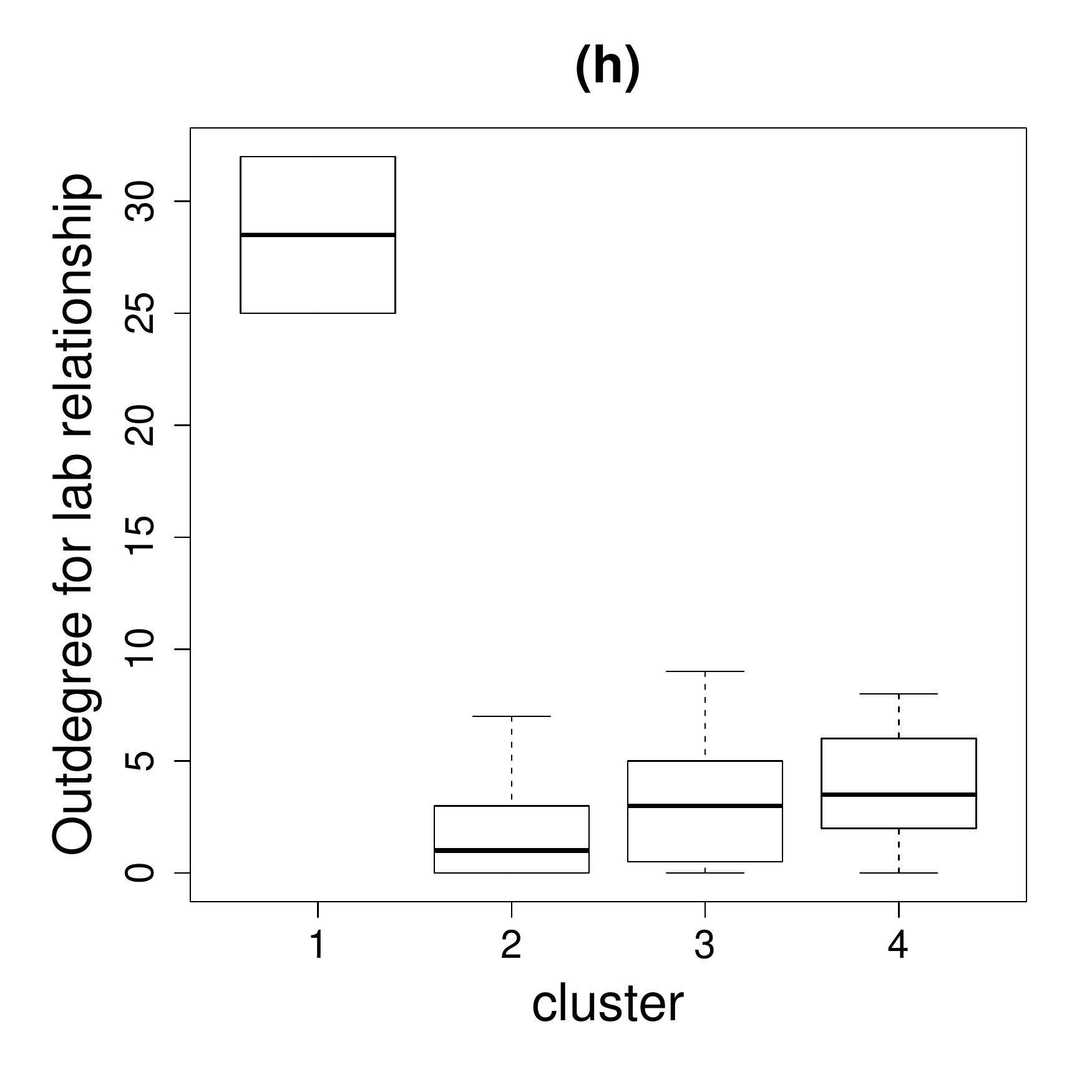}\\
\end{tabular}
  \caption{Boxplots of age (a),
  lab size (b),
  performance in 1st period (c),
  performance in 2nd period (d),
  indegree for researcher relationship (e),
  outdegree for researcher relationship (f),
  indegree for lab relationship (g) and
  outdegree for lab relationship (h)
  grouped by blocks.}
  \label{figboxplot}
\end{figure}

 \section{Discussion}\label{sec=discussion}

The essence of ‘networks’ is to help actors  cut across organizational boundaries to create new relationships \citep{baker1992network},
to identify new opportunities and, eventually, to create new organizations to use or
hoard these new opportunities \citep{lazega2012analyses,lazega2013network,tilly1998durable}. This work is a step forward a more precise comprehension of such mechanisms.

In this paper, we proposed a SBM  modelisation for multiplex networks.  This way of partitioning the graph and calculating the probabilities that two individuals are connected
takes into account more parameters in the graph than just centrality scores and size as in \citet{lazega2007poissons}.
Besides, compared to MERGMs \citep{wang2013exponential},
this method takes inhomogeneities of actors more into account because the separate blocks tend to make sense
as  blocks with a specific identity. Multiplex SBM  statistical analysis designed for the analysis of multilevel networks given as
``superposed'' \citep{lazega2007poissons} network data aims at identifying the rules that govern the formation
of links at the inter-individual level conditioned by the characteristics of the inter-organizational network.
Our analysis found a general structure that shows the ways in which these rules can be derived in this case and
help assess actor inhomogeneities in terms of their influence on parameter estimates. This assessment is new in
the sense that it shows that formation of connections in the inter-individual network does not apply to all actors in
the network identically.

In this paper, the model is applied on $K=2$ networks.  Applying with $K$ much larger would imply estimation difficulties since the number of parameters  exponentially depends  on $K$. However, particular assumption  (such as Markovian dependencies) can reduce the difficulty and lead to feasible cases.
Besides, in our example, going from multilevel to multiplex was quite natural since few laboratories contained more than one researcher. If the number of organizations were far smaller then the number of individuals than maybe it would be interesting to introduce an other kind of relation,
specifying if the individuals are in the same laboratory or not. In this case, the probabilities vector would have a special structure, taking this specificity into account (if the individuals belong to the same organization then they are automatically linked through their organization) .

In this version of the analysis, the covariates were studied a posteriori, the classification being purely done on the network. From a practical point of view, including the covariates in the model would imply estimation difficulty (and even more if 
the effect of the covariates on the connexion probabilities depends on the blocks). However, even if attractive at first sight, the choice of including   the covariates in the modeling has to be questioned.
Including them will cluster the data beyond the effect of these covariates, while not including them
 will provide a description of  blocks on the basis of the relation, may this relation be influenced by these covariates. This second strategy may lead to more interpretable results.

The first aim of \cite{Lazega2008} was cancer research management. Thus,   the covariate ``performance in terms of publication'' catches more attention.  Indeed, this modelisation can generate other questions such as the influence of the networks on the performance of the actors (publications in our case). In this work, the performance is treated as a factor to explain the relations. However, one could dream of a model which would help the actors  build the ideal network to optimize the performance.  If it is true that contemporary society is an ``organizational society''
\citep{coleman1982asymmetric,perrow1991society,presthus1962organizational} -in the sense that action and
performance measured at the individual level strongly depend on the capacity of the actor to construct and
to use organizations as instruments, and thus to manage his/her interdependencies at different levels in a strategic manner-,
then the study of interdependencies jointly at the inter-individual and the inter-organizational level is
important for numerous sets of problems. Proposing a hierarchical  model in this direction is out of the scope of the paper but is clearly a possible extension.

To conclude, application to sociological network analysis will clearly benefit from this methodological contribution.
In spite of several limitations listed before,  this analysis of multiplex networks seems
therefore adapted to certain types of questions that sociologists ask when they try to combine both individual and
contextual factors in order to estimate the likelihood of an individual or a group to adopt a given behavior or to reach a
given level of performance. More generally, this approach explores a complex meso-social level of accumulation,
of appropriation and of sharing of multiple resources. This level, still poorly known, is difficult to observe
without a structural approach.
\\

\newpage
\appendix
\section{Identifiability of the multiplex SBM and convergence of the maximum likelihood estimates} \label{app:identif}

\subsection{Identifiability}

\cite{celisse:daudin:laurent:2012} have proved the identifiability of the parameters for the uniplex Bernoulli Stochastic Block models. We extend   their results to  the multiplex SBM.
First we  recall the notations : $\forall (i,j) \in \{1,\dots,n\}^2$, $i\neq j$, $\forall (q,l) \in \{1,\dots,Q\}^2$,$\forall w \in \{0,1\}^K$, we set:
\begin{eqnarray*}
\pi^{(w)}_{ql} &= & \P(X^{1:K}_{ij}=w |Z_i=q,Z_j=l)\\
\alpha_q& =& P(Z_i = q)
\end{eqnarray*}

Let  us introduce $\bpi^{(w)} = (\pi_{ql}^{(w)})_{(q,l) \in \{1 \dots Q\}}$ and $\bpi = (\bpi^{(w)})_{w \in \{0,1\}^K}$.  $\balpha = (\alpha_1, \dots,\alpha_Q)$. Theorem \ref{theo:identif} sets the identifiability of $\theta=(\balpha,\bpi)$ in  the  multiplex SBM.

 \begin{theorem}\label{theo:identif}
Let $n\geq 2Q$.  Assume that for any $q \in \{1, \dots Q\}$, $\alpha_q>0$ and for every $w \in \{0,1\}^K$, the coordinates of $ \boldsymbol{r}^{(w)}(\theta) = \bpi^{(w)} \cdot \balpha$ are distinct. Then, the multiplex SBM parameter $\theta = (\balpha,\bpi)$ is identifiable.
 \end{theorem}
 \begin{proof}
 The proof given by  \cite{celisse:daudin:laurent:2012} can be directly extended to our case leading to the expected result.
 As explained in the original paper, the identifiability can be proved using algebraic tools.
 The only difference in the multiplex context is that the proof has to be applied for any type of tie $w \in \{0,1\}^K$.

For any  $w \in \{0,1\}^Q$, we set  $r_q^{(w)}(\theta)$ the probability for a member of block $q$ to have a tie of type $w$ with an other individual: $r_q^{(w)}(\theta) = \sum_{l=1}^Q \pi^{(w)}_{ql} \alpha_l $.  Let $R^{(w)}(\theta)$ be the $Q$-square matrix such that  $ R_{iq}^{(w)}(\theta) = (r_{q}^{(w)}(\theta))^i$  for $i=0\dots Q-1$ and $q=1\dots Q$.   $R^{(w)}$ is a Vandermonde matrix, which is invertible by assumptions on the coordinates of  $ \boldsymbol{r}^{(w)}(\theta)$. Now, for any $w$ and  any $i=0 \dots 2Q-1$,  we set:
 \begin{equation}\label{eq:u}
 u_{i}^{(w)}(\theta) = \P(X^{1:K}_{11}= \dots = X^{1:K}_{1i} = w; \theta) =  \sum_{q=1}^Q \alpha_q (r^{(w)}_q(\theta))^i\,,
 \end{equation}
 and $M^{(w)}(\theta)$ is a $(Q+1) \times Q$ matrix such that
  \begin{equation}\label{eq:M}
  M^{(w)}_{ij}(\theta) = u_{i+j}^{(w)}(\theta), \quad i=0\dots Q, j=0\dots Q-1\,.
  \end{equation}
 For any $i=0,\dots,Q$ we define  the $Q$-square matrix $M^{i(w)}(\theta)$ by removing line $i$ from this matrix. In particular,
 \begin{equation}\label{eq:M2}
 M^{Q(w)}(\theta) = R^{(w)}(\theta)A(\theta) R^{(w)}(\theta)^{\intercal}\,,
 \end{equation} where $A(\theta)$ is the $\balpha$-diagonal matrix.
All the $\alpha_q$ being non-null and $R^{(w)}(\theta)$ being invertible, then $\det(M^{Q(w)})>0$.
So, if we define:
\begin{equation}
 \nonumber
 B(X;\theta) = \sum_{i=0}^{Q} (-1)^{i+Q} \det(M^{i(w)}(\theta))X^i\,.
\end{equation}
$B$ is of degree $Q$.
Let us define $V^{i(w)}(\theta)=(1,r_i^{(w)}(\theta),  \dots ,(r_i^{(w)}(\theta))^Q)$, then
\begin{equation}
 \nonumber
 B(r_i^{(w)}(\theta);\theta) = \det \left(M^{(w)}(\theta),V_i^{(w)}(\theta)\right)\,,
\end{equation}
 where  $\left(M^{(w)}(\theta),V_i^{(w)}(\theta)\right)$ is a $(Q+1)$ square matrix. The columns of $M$ being linear combinations of the $V_i^{(w)}$, we obtain  $B(r_i^{(w)}(\theta);\theta)=0$ for any $i = 0 \dots Q-1$. So, we can factorize $B$ as
 \begin{equation}\label{eq:B}
 B(x;\theta) =\det(M^{Q(w)}(\theta)) \prod_{i=0}^{Q-1}(x-r_i^{(w)}(\theta))\,.
 \end{equation}

 Now assume that $\theta = (\bpi,\alpha) $ and $\theta'=(\bpi',\alpha')$ are two sets or parameters such that  for any multiplex graph $\Xall$,
 $\ell(\Xall; \theta) = \ell(\Xall; \theta')$. Consequently,  from equation (\ref{eq:u}),
 we get $ u_{i}^{(w)}(\theta) =  u_{i}^{(w)}(\theta')$ and from equation (\ref{eq:M}), we deduce $M^{i(w)}(\theta) = M^{i(w)}(\theta')$ for any $i=0, \dots, Q-1$.  The polynomial function $B$ depending on the determinant of the $M^{i(w)}(\theta)$'s, we have $B(\cdot; \theta) = B(\cdot; \theta')$,
 leading to $r_i{(w)}(\theta) = r_i{(w)}(\theta')$ for all $i=0, \dots, Q-1$ thanks to equation (\ref{eq:B}). Thus $R^{(w)}(\theta) = R^{(w)}(\theta')$ and
\begin{equation}
\nonumber
A(\theta) = (R^{(w)}(\theta)^{\intercal})^{-1}  M^{Q(w)}(\theta)  R^{(w)}(\theta)  =A(\theta')\,.
\end{equation}
As a consequence, $\balpha = \balpha'$. \\
 Finally, let $U^{(w)}_{i,j}$ ($0\geq i,j,\geq Q$) denote :
\begin{equation}
\nonumber
U^{(w)}_{i,j} = \P(X^{1:K}_{1,k} = w, k= 1, \dots, i+1  \mbox{ and } k = n-j+1, \dots,n)\,. 
\end{equation}
We can write the $Q\times Q$ matrix $U^{(w)}(\theta)$ as  $U^{(w)}(\theta) = R^{(w)}(\theta) A(\theta) \bpi^{(w)} A(\theta) (R^{(w)}(\theta))^{\intercal}$.
 Using the fact that $ R^{(w)}(\theta) = R^{(w)}(\theta') $ and $A(\theta)  =  A(\theta')$,
 we  get
 \begin{equation}
  \nonumber
U^{(w)}(\theta)  = U^{(w)}(\theta') \Rightarrow  \bpi^{(w)} =  (\bpi'^{(w)}),\quad \forall w \in \{0,1\}^K\,.  
 \end{equation}
And the theorem is demonstrated.

\end{proof}

\subsection{Consistency  of the maximum likelihood estimator in multiplex SBM models}

We study the asymptotic properties of the maximum likelihood estimator in multiplex SBM models. Let $\A_1$, $\A_2$ and $\A_3$ be   the following assumptions: \\
 \emph{Assumption $\A_1$}: for every $q \neq q'$, there exists $l \in \{1, \dots ,Q\}$ such that $  \boldsymbol{\pi}_{ql} \neq \boldsymbol{\pi}_{q'l}$, or $   \boldsymbol{\pi}_{lq} \neq \boldsymbol{\pi}_{lq'} $. \\
 \emph{Assumption $\A_2$}: there exists $\zeta>0$  such that $\forall (q,l) \in {1,\dots,Q}$, $\pi^{(w)}_{ql} \in]0,1[ \Rightarrow \pi^{(w)}_{ql} \in [\zeta, 1-\zeta]$.\\
  \emph{Assumption $\A_3$}: there exists $\gamma \in \, (0,1/Q)$  such that $\forall q \in {1,\dots,Q}$, $\alpha_q \in [\gamma, 1-\gamma]$.\\

Let $\Xall$ be a realization  from the multiplex SBM: $\mathbb{P}(\; \cdot \;  ; \bZ^*,\balpha^*,\bpi^*)$ where $\bZ^*=(Z^*_1,\dots,Z^*_n)$ is the true group label  sequence  and $(\balpha^*,\bpi^*)$ are  the true parameters.  Let $(\widehat{\balpha},\widehat{\bpi}) $ be the maximum likelihood estimator defined as:
\begin{equation}
 \nonumber
(\widehat{\balpha},\widehat{\bpi})  =  \mbox{Argmax}_{(\balpha,\bpi)} \mathcal L(\Xall; \balpha,\bpi)\,,
 \end{equation}
where
\begin{eqnarray*}
\ell(\Xall; \balpha,\bpi) &=& \sum_{\bZ\in\{1,\dots Q\}^n} e^{\ell_1(\Xall | \bZ;\bpi)} \prod_{i=1}^n \alpha_{Z_i}\,,  \\
 \ell_1(\boldsymbol{ X }| \bZ; \bpi)&=& \sum_{i,j, i\neq j}^n \log \P(X^{1:K}_{ij} | Z_i, Z_j;\bpi)=\sum_{i,j, i\neq j}^n \sum_{w \in \{0,1\}^K} \ind_{X^{1:K}_{ij}=w} \log \pi^{(w)}_{Z_i Z_j}\,.
 \end{eqnarray*}

 \begin{theorem}\label{theo:consist}
  Let ($\A_1$), ($\A_2$) and ($\A_3$) hold. Then, for any distance $d(\cdot, \cdot)$ on the set of  parameter $\bpi$,  we have:
 \begin{equation}
  \nonumber
d(\widehat{\bpi} , \bpi^*) \xrightarrow[n \rightarrow \infty]{\mathbb{P}} 0\,.  
 \end{equation}
 Moreover, assume that $\|\widehat{\bpi} - \bpi^*\|_{\infty} = o_{\mathbb{P}}(\sqrt{\log(n)}/n)$  then,
  for any distance $d(\cdot,\cdot)$  in $\mathbb{R}^Q$,
\begin{equation}
 \nonumber
 d(\widehat{\balpha} , \balpha^*)\xrightarrow[n \rightarrow \infty]{\mathbb{P}} 0\,.
\end{equation}
      \end{theorem}
Note the rate $o_{\mathbb{P}}(\sqrt{\log(n)}/n)$ has not been proved yet. However, in the unilevel context, there is empirical evidence that the rate convergence on $\widehat{\bpi}$ is $1/n$ \citep{gazal2012accuracy}.

\section{Variational EM algorithm for multiplex SBM : principle, details and convergence}\label{appendix_VarEM}

\subsection{General principal of the variational EM}

The Stochastic Block Models belong to the  incomplete data models class,
the non-observed data being the block indices $(Z_{i})_{i=1\dots n} \in \{1, \dots, Q\}^n $. As written before, the likelihood has a marginal expression:
\begin{equation}\label{eq:vraismarg}
\ell(\Xall ; \theta) =   \sum_{\bZ\in\{1,\dots Q\}^n} e^{\ell_1(\Xall | \bZ;\bpi)}p(\bZ; \balpha)\,,
\end{equation}
which is  not tractable as soon as $n$ and $Q$ are large. The variational EM is an alternative method  to maximize the marginal likelihood with respect to $\theta$.  The variational EM (applied in the SBM context by \cite{Daudinetal2008})  relies on the following decomposition of (\ref{eq:vraismarg}). Let $\mathcal{R}_{\Xall}$ be any probability  distribution on $\Zall$, we have:
\begin{eqnarray}\label{eq:vraismarg decomp}\nonumber
\log \ell(\Xall | \theta) &=& \sum_{\bZ} \mathcal{R}_{\Xall}(\bZ) \log p(\Xall, \bZ; \theta)-  \sum_{\Zall} \mathcal{R}_{\Xall}(\bZ)  \log \mathcal{R}_{\Xall}({\bZ})\\
&&+  \mathbf{KL}[\mathcal{R}_{\Xall}, p(\cdot | \Xall; \theta)]\,,
\end{eqnarray}
where $\mathbf{KL}$ is the Kullback-Leibler distance, $ p(\Xall, \bZ; \theta)$ is the joint density of $\Xall$ and $\bZ$ (namely the complete likelihood)   and $p(\cdot | \Xall; \theta)$ is the posterior distribution of $\Zall$ given  the data $\Xall$ and the parameters $\theta$. Instead of maximizing $\log \ell(\Xall; \theta) $, the variational EM optimizes a lower bound $\mathcal{I}_{\theta}(\mathcal{R}_{\Xall})$  of  $\log \ell(\Xall | \theta)$ where:
\begin{eqnarray}
\nonumber
\mathcal{I}_{\theta}(\mathcal{R}_{\Xall})  &=& \log   \ell(\Xall | \theta)  -  \mathbf{KL}[\mathcal{R}_{\Xall}, p(\cdot | \Xall; \theta)]\,, \\
&=&    \sum_{\bZ} \mathcal{R}_{\Xall}(\bZ) \log p(\Xall, \bZ; \theta) -  \sum_{\Zall} \mathcal{R}_{\Xall}(\bZ)  \log \mathcal{R}_{\Xall}({\bZ}) \label{eq:decomp VarEM}\,,\\
\nonumber
&\leq& \log \ell(\Xall | \theta)\,.
\end{eqnarray}
Note that, thanks to equality (\ref{eq:decomp VarEM}), optimizing $\mathcal{I}_{\theta}(\mathcal{R}_{\Xall})$ with respect to $\theta$ 
no longer requires the computation of the marginal likelihood.
Note also that the equality $\mathcal{I}_{\theta}(\mathcal{R}_{\Xall})  = \log \ell(\Xall; \theta)$ holds if and only 
if  $\mathcal{R}_{\Xall} = p(\cdot | \Xall; \theta)$. As a consequence, $\mathcal{R}_{\Xall}$ will be taken as an approximation
of $p(\cdot | \Xall; \theta)$ in a certain class of distributions. \citet{jaakkola00} proposed to optimize it in the following class:
\begin{equation}
 \nonumber
\mathcal{R}_{\Xall,\btau}(\Zall) = \prod_{i=1}^n h(Z_i, \widehat{\btau}_i)\,,
 \end{equation}
where $h(\cdot; \btau_i)$ is the multinomial distribution of parameter $\btau_i = (\tau_{i1}, \dots, \tau_{iq})$.
 Finally, the variational EM updates alternatively $\theta$ and $\btau$ in the following way. At iteration $(t)$, given the current state $(\theta^{(t-1)},\btau^{(t-1)})$,
\begin{enumerate}
\item[-]\textbf{Step 1} Compute \\$\btau^{(t)} = \arg \min_{\btau}  \mathbf{KL}[\mathcal{R}_{\Xall,\tau}, p(\cdot | \Xall; \theta^{(t-1)})] = \arg \max_{\btau} \mathcal{I}_{\theta^{(t-1)}}(\mathcal{R}_{\Xall,\btau})  $.
\item[-]\textbf{Step 2} Compute $\theta^{(t)} = \arg \max_{\theta}   \mathcal{I}_{\theta^{(t)}}(\mathcal{R}_{\Xall, \btau^{(t)}}) $.
\end{enumerate}
The details of steps \textbf{1} and  \textbf{2} directly depend on the considered statistical model. For uniplex SBM without covariates,  they are given in \cite{Daudinetal2008}. The details for the multiplex SBM are given here after.

\vspace{1em}

\subsection{Details of the calculus for multiplex SBM models}
We now  detail \textbf{Step 1}  and \textbf{Step 2} for multiplex SBM models.

$\bullet$ \textbf{Step 1}:   $\btau^{(t)}$ verifies 
\begin{equation}
 \nonumber
\btau^{(t)} = \arg \min_{\btau}  \mathbf{KL}[\mathcal{R}_{\Xall,\tau}, p(\cdot | \Xall; \theta^{(t-1)})] = \arg \max_{\btau} \mathcal{I}_{\theta^{(t-1)}}(\mathcal{R}_{\Xall,\btau}) \,.
 \end{equation}

We first rewrite  $\mathcal{I}_{\theta}(\mathcal{R}_{\Xall,\btau})$ for this special context:
\begin{eqnarray*}
 \mathcal{I}_{\theta}(\mathcal{R}_{\Xall,\btau}) &=&   \sum_{\bZ} \mathcal{R}_{\Xall,\btau}(\bZ) \log p(\Xall, \bZ; \theta) -  \sum_{\Zall} \mathcal{R}_{\Xall,\btau}(\bZ)  \log \mathcal{R}_{\Xall,\btau}({\bZ})\,,
 \end{eqnarray*}
 with
 \begin{eqnarray*}
 \log p(\Xall, \bZ; \theta) &=&  \ell_1(\Xall |  \bZ; \theta)+ \log  p(\bZ; \theta)\,,\\
 %&=&   \log \prod_{i,j,i\neq j}p(X^{1:K}_{ij} |  Z_i,Z_j; \theta) + \log \prod_{i=1}^n \alpha_{Z_i}\,,\\
 &=& \sum_{i,j, i\neq j}  \log p(X^{1:K}_{ij} |  Z_i,Z_j; \theta)  + \sum_{i=1}^n \log\alpha_{Z_i}\,.
 \end{eqnarray*}
 This quantity has to be integrated over $\bZ$ where $\bZ \sim \mathcal{R}_{\Xall,\btau}$ which means that $ \bZ= (Z_i)_{i=1\dots n} $ are independent variables such that $\P(Z_{i}=q) = \tau_{iq}$. We obtain:
 \begin{eqnarray*}
 \mathcal{I}_{\theta}(\mathcal{R}_{\Xall,\btau}) &=& \sum_{\bZ} \mathcal{R}_{\Xall,\btau}(\bZ) \left[ \sum_{i,j, i\neq j}  \log p(X^{1:K}_{ij} |  Z_i,Z_j; \theta)  + \sum_{i=1}^n \log\alpha_{Z_i}\right] \\
 &&-  \sum_{\Zall} \mathcal{R}_{\Xall,\btau}(\bZ)  \log \mathcal{R}_{\Xall,\btau}({\bZ})\,,\\
 &=&  \sum_{q,l} \sum_{i,j, i\neq j}  \log p(X^{1:K}_{ij} |  Z_i=q,Z_j=l; \theta)\tau_{iq} \tau_{jl}  \\
 &&+ \sum_{i=1}^n \sum_{q=1}^Q  \tau_{iq}  \log\alpha_{q}  - \sum_{i=1}^n \sum_{q=1}^Q \tau_{iq} \log \tau_{iq}\,, \\
 \end{eqnarray*}
 where $\log p(X^{1:K}_{ij} |  Z_i=q,Z_j=l; \theta)$'s expression is given in equation (\ref{modelK}).

 $\mathcal{I}_{\theta}(\mathcal{R}_{\Xall,\btau}) $ has to be maximized with respect to $\btau$ under the constraint: $\forall i =1 \dots n$, $\sum_{q=1}^Q \tau_{iq}=1$.
 As a consequence, we compute the derivatives of  $\; \mathcal{I}_{\theta}(\mathcal{R}_{\Xall,\btau})  + \sum_{i=1}^n \lambda_i \left[ \sum_{q=1}^Q \tau_{iq}-1\right]$ with respect to $(\lambda_i)_{i=1\dots n}$ and $(\tau_{iq})_{i=1\dots n, q=1 \dots Q}$ where $\lambda_i$ are the Lagrange multipliers, leading to the following collection of equations:  for $i=1\dots n$ and $q=1\dots Q$,

 \begin{equation}
  \nonumber
 \sum_{l} \sum_{j=1, j\neq i}^n  \log p(X^{1:K}_{ij} |  Z_i=q,Z_j=l; \theta)\tau_{jl}   +   \log\alpha_{q}  -   \log \tau_{iq} +1 + \lambda_i=0\,,
 \end{equation}
 which leads to the following fixed point problem:

 \begin{equation}
  \nonumber
  \widehat{\tau}_{iq} = e^{1+ \lambda_i} \alpha_q \prod_{j=1, j\neq i}^n \prod_{l=1}^Q p(X^{1:K}_{ij} |  Z_i=q,Z_j=l; \theta) ^{\widehat{\tau}_{jl}}, \quad \forall  i =1 \dots n,  \forall  q=1 \dots Q\,,
 \end{equation}
 which has to be solved under the constraints $\forall i =1 \dots n$, $\sum_{q=1}^Q \tau_{iq}=1$. This optimization  problem is solved using a  standard fixed point algorithm.

\textbf{Step 2} Compute $\theta^{(t)} = \arg \max_{\theta}   \mathcal{I}_{\theta^{(t)}}(\mathcal{R}_{\Xall, \btau^{(t)}}) $.

Once the $\widehat{\btau}$ have been optimized, the parameters $\theta$ maximizing $\mathcal{I}_{\theta}(\mathcal{R}_{\Xall, \widehat{\btau}})$ have to be computed under the constraints:
$ \sum_{q=1}^Q \alpha_q=1$ and  $\sum_{w \in \{0,1\}^L} \pi_{ql}^{(w)} = 1$ for all $(q,l) \in \{1, \dots,Q\}^2$.

The maximization with respect to $\balpha$ is quite direct and in any case, we obtain:
\begin{eqnarray*}
\widehat{\alpha}_q &=& \frac{1}{n} \sum_{i=1}^n \widehat{\tau}_{iq},  \quad  \widehat{\pi}_{ql}^{(w)} = \frac{\sum_{ij} \widehat{\tau}_{iq}  \widehat{\tau}_{jl} \ind_{X^{1:K}_{ij}=w}}{\sum_{ij} \widehat{\tau}_{iq}  \widehat{\tau}_{jl} }
 \end{eqnarray*}
\begin{Rem}
If the edge probabilities depend on covariates: 
\begin{equation}
 \nonumber
 \mbox{logit}(\pi_{ql}^{(w)}) = \mu_{ql}^{(w)} +  (\beta_{ql}^{(w)})^{\intercal} \mathbf{y}_{ij}\,,
\end{equation}
then the optimization of $ (\mu_{ql}^{(w)})$ and $ (\beta_{ql}^{(w)})$ at step 2 of the VarEm is not explicit anymore and one should resort  to optimization algorithms such as Newton-Raphson
algorithm. \end{Rem}

\subsection{Convergence of the VarEM estimates}

We now consider the consistency of the estimates obtained by the variational EM algorithm.
 Using the variational EM previously described  is equivalent to maximizing the so-called \emph{variational-likelihood} $\mathcal{I}$ where:
\begin{equation}
 \nonumber
 \mathcal{I}(\Xall; \btau,\balpha,\bpi) = \sum_{i \neq j} \tau_{iq}\tau_{jl}  \sum_{w \ \{0,1\}^K} \ind_{X_{ij}=w} \log \pi^{(w)}_{ql} - \sum_{iq} \tau_{iq}(\log \tau_{iq}-\log\alpha_q)\,.
\end{equation}
The variational estimators (VE) are obtained by:
\begin{eqnarray*}
 \widetilde{\btau}(\balpha,\bpi) &=& \arg\max_{\btau} \;  \mathcal{I}(\Xall; \btau,\balpha,\bpi) \quad  (\widetilde{\balpha},\widetilde{\bpi})= \arg\max_{\balpha,\bpi} \;  \mathcal{I}(\Xall; \widetilde{\btau},\balpha,\bpi)
  \end{eqnarray*}
 \begin{theorem}\label{theo:conv var EM}
Assume that ($\A_1$),  ($\A_2$) and ($\A_3$) hold. Then for any distance on the set of parameters $\bpi$,
\begin{equation}
 \nonumber
d(\widetilde{\bpi},\bpi^*) \xrightarrow[n \rightarrow \infty]{\mathbb{P}} 0\,.
 \end{equation}
Moreover, assume that  $d(\widetilde{\bpi},\bpi^*)=o_{\P}(1/n)$, then for any distance on $\mathbb{R}^Q$,
\begin{equation}
 \nonumber
 d(\widetilde{\balpha},\balpha^*) \xrightarrow[n \rightarrow \infty]{\mathbb{P}} 0\,.
\end{equation}
\end{theorem}

\subsection{About the proofs}

   The proofs of these results require  many intermediate results which won't be given here because their adaptation to the multiplex context is quite direct.
   Indeed,  at any step of the proof, the original Bernoulli distribution needs to replaced by its $K$-dimensional  version. More precisely, a central quantity in the proof is the ratio  $\log \frac{p( \bZ | \Xall; \balpha,\bpi)}{p( \bZ^* | \Xall; \balpha,\bpi)}$.
   In the unilevel case, this quantity is:
     \begin{eqnarray*}
\log \frac{p( \bZ  | \Xall; \balpha,\bpi)}{p( \bZ^* | \Xall; \balpha,\bpi)}  &=&\log \frac{p(\Xall|\bZ ;\bpi)}{p(\Xall|\bZ^*;\bpi)} + \log\frac{p(\bZ; \balpha)} {p(\bZ^*; \balpha)}\\
&=& \sum_{i\neq j} X_{ij}\log \frac{\pi^{(w)}_{Z_i,Z_j}}{\pi^{(w)}_{z^*_i,z^*_j}} + (1-X_{ij}) \log \frac{1-\pi^{(w)}_{Z_i,Z_j}}{1-\pi^{(w)}_{z^*_i,z^*_j}}  + \sum_{i=1}^n \log \frac{\alpha_{Z_i}}{\alpha_{z^*_i}}
\end{eqnarray*}
In the multiplex case, the sum of two terms is replaced by a sum of $2^K$ terms:
   \begin{eqnarray*}
\log \frac{p( \bZ | \boldsymbol{X}; \balpha,\bpi)}{p( \bZ^* | \boldsymbol{X}; \balpha,\bpi)}
&=& \sum_{i\neq j} \sum_{w \in \{0,1\}^K} \ind_{\{X_{ij}=w\}} \log \frac{\pi^{(w)}_{Z_i,Z_j}}{\pi^{(w)}_{Z^*_i,Z^*_j}} + \sum_{i=1}^n \log \frac{\alpha_{Z_i}}{\alpha_{z^*_i}}
\end{eqnarray*}
where $\sum_{w \in \{0,1\}^K}  \pi_{ql}^{(w)}=1$, for any $(q,l)$. Going from two terms to $2^K$ terms does not imply any mathematical difficulty and so does not compromise the convergence results.

 \bibliographystyle{apalike}
\bibliography{multilevel_multiplex_network}

\end{document}